\documentclass[orivec]{llncs}

\usepackage{amssymb}
\usepackage{import}
\usepackage{catchfilebetweentags}
\usepackage{multirow}
\usepackage{complexity}
\usepackage{amsfonts}
\usepackage[title,titletoc,toc]{appendix}
\usepackage{scalefnt}
\usepackage{adjustbox}
\usepackage{afterpage}

\usepackage{colortbl}
\usepackage{xcolor}

\usepackage[nounderscore]{syntax}
\usepackage[perpage,symbol,stable]{footmisc}
\usepackage{newfloat}

\usepackage{paralist}
\usepackage{mdwlist}

\makecompactlist{itemize}{stditemize}

\makecompactlist{description}{stddescription}

\usepackage[latin9]{inputenc}
\usepackage{listings}
\usepackage{float}
\usepackage{booktabs}
\usepackage{url}
\usepackage{amsmath}
\usepackage{amssymb}
\usepackage{xargs}[2008/03/08]
\usepackage[T1]{fontenc}
\usepackage{lmodern}

\usepackage{tikz}
\usepackage{wrapfig}
\usepackage{expl3}
\usepackage{rotating}
\usepackage{xspace}
\usepackage{colonequals}
\usepackage{tabularx}
\usepackage{cleveref}
\usepackage{etoolbox}

\csundef{enumerate*}
\csundef{endenumerate*}
\csundef{itemize*}
\csundef{enditemize*}
\csundef{description*}
\csundef{enddescription*}
\usepackage[inline]{enumitem} 
\setlist[enumerate]{label=(\roman{*})}

\usepackage{environ}
\makeatletter
\newsavebox{\measure@tikzpicture}
\NewEnviron{scaletikzpicturetowidth}[1]{%
  \def\tikz@width{#1}%
  \def\tikzscale{1}\begin{lrbox}{\measure@tikzpicture}%
  \BODY
  \end{lrbox}%
  \pgfmathparse{#1/\wd\measure@tikzpicture}%
  \edef\tikzscale{\pgfmathresult}%
  \BODY
}
\makeatother

\usetikzlibrary{backgrounds,petri}
\usetikzlibrary{patterns} 
\usetikzlibrary{fit}
\usetikzlibrary{automata}
\usetikzlibrary{trees}
\usetikzlibrary{shapes.geometric}

\usepackage{pgfplots}
\usepackage{tikz-3dplot}
\usepackage{pifont}

\usepackage{caption}
\makeatletter

\newcommand{\localCommand}{} 

\newcommand{\newlengthsettowidth}[2]{\newlength {#1} \settowidth{#1}{#2}}
\newcommand{\newcounterset}      [2]{\newcounter{#1} \setcounter{#1}{#2}}

\newcommand{\ensurecommand}[2]{
  \providecommand{#1}{}
  \renewcommand{#1}{#2}}

\newcommand{\atl}{\geq}                                 
\newcommand{\atm}{\leq}                                 

\ensurecommand{\subsetneq}  {\subsetneqq}               

\newcommand{\angles}  [1]{     \langle #1       \rangle} 

\newcommand{\Ifthenelse}[3]{\ifthenelse{#1}{#2}{#3}}   
\newcommand{\Ifthen}    [2]{\Ifthenelse{#1}{#2}{}}

\newcommand{\Equal}     [2]{\equal{#1}{#2}}            
\newcommand{\Empty}     [1]{\Equal{#1}{}}

\newlength{\posLength}
\newcommand{\pos}[3][c]{\settowidth{\posLength}{#3}\makebox[\posLength][#1]{#2}}
\newlengthsettowidth{\tabLength}{\ \ \ }

\newcommand{\wbox}[2][\tabLength]{\hspace*{#1}\mbox{#2}\hspace*{#1}}

\newcommand{\Paragraph}[1][\baselineskip]{\vspace{#1}}
\newcommand{\End}{\end{document}}                       
\newcommand{\format}{}                                  
\newcommand{\emphasize}[1]{\textbf{#1}}                 
\newcommand{\emphdef}[1]{\emphasize{#1}}                

\newtheorem{DEF}{Definition}
\crefname  {DEF}{Definition}{Definitions}
\Crefname  {DEF}{Definition}{Definitions}
\newtheorem{THE}[DEF]{Theorem}
\crefname  {THE}{Theorem}{Theorems}
\Crefname  {THE}{Theorem}{Theorems}
\newtheorem{LEM}[DEF]{Lemma}
\crefname  {LEM}{Lemma}{Lemmata}
\Crefname  {LEM}{Lemma}{Lemmata}
\newtheorem{COR}[DEF]{Corollary}
\crefname  {COR}{Corollary}{Corollaries}
\Crefname  {COR}{Corollary}{Corollaries}

\crefname  {EXA}{Example}{Examples}
\Crefname  {EXA}{Example}{Examples}

\crefname  {PRO}{Property}{Properties}
\Crefname  {PRO}{Property}{Properties}

\crefname  {figure}{Figure}{Figures}
\Crefname  {figure}{Figure}{Figures}

\crefname  {algorithm}{Algorithm}{Algorithms}
\Crefname  {algorithm}{Algorithm}{Algorithms}

\crefname  {table}{Table}{Tables}
\Crefname  {table}{Table}{Tables}

\crefname  {section}{Section}{Sections}
\crefname  {subsection}{Section}{Sections}

\crefname  {appendix}{Appendix}{Appendices}

\usepackage{listings}
\@ifundefined{showcaptionsetup}{}{%
 \PassOptionsToPackage{caption=false}{subfig}}
\usepackage{subfig}
\makeatother

\tikzstyle{bwnode}=[] 
\tikzstyle{minstate}=[] 
\tikzstyle{fwnode}=[] 
\tikzstyle{maxstate}=[] 

\ExplSyntaxOn
\newcommand{\SHAREDSTATE}[1]{}

\newcommand{\LOCALSTATE}[1]{
\int_to_symbols:nnn {#1} {18202224}%
{%
{0}{\astate{0}{0}{0}{\labela}}%
{1}{\astate{0}{0}{1}{\labela}}%
{2}{\astate{0}{1}{0}{\labela}}%
{3}{\astate{0}{1}{1}{\labela}}%
{4}{\astate{1}{0}{0}{\labela}}%
{5}{\astate{1}{0}{1}{\labela}}%
{6}{\astate{1}{1}{0}{\labela}}%
{7}{\astate{1}{1}{1}{\labela}}%
{8}{\astate{0}{0}{0}{\labelb}}%
{9}{\astate{0}{0}{1}{\labelb}}%
{10}{\astate{0}{1}{0}{\labelb}}%
{11}{\astate{0}{1}{1}{\labelb}}%
{12}{\astate{1}{0}{0}{\labelb}}%
{13}{\astate{1}{0}{1}{\labelb}}%
{14}{\astate{1}{1}{0}{\labelb}}%
{15}{\astate{1}{1}{1}{\labelb}}%
{16}{\astate{0}{0}{0}{\labelc}}%
{17}{\astate{0}{0}{1}{\labelc}}%
{18}{\astate{0}{1}{0}{\labelc}}%
{19}{\astate{0}{1}{1}{\labelc}}%
{20}{\astate{1}{0}{0}{\labelc}}%
{21}{\astate{1}{0}{1}{\labelc}}%
{22}{\astate{1}{1}{0}{\labelc}}%
{23}{\astate{1}{1}{1}{\labelc}}%
{146}{\labela/(\PREDa\ne\PREDc)\andop(\PREDc\impliesop\negop\PREDb)}
{172022}{\labelc/(\PREDa\ne\PREDc)\andop(\PREDc\impliesop\negop\PREDb)}
{18202223}{\labelc/(\negop\PREDc\andop(\PREDb\orop\PREDa)) \orop (\PREDa \andop \PREDb)}
{182022}{\labelc/\negop\PREDc\andop(\PREDa\orop\PREDb)}
{2022}{\labelc/\PREDa\andop\negop\PREDc}
{2456}{\labela/(\PREDa\andop\negop\PREDb)\orop(\PREDb\andop\negop\PREDc)}
}%
}
\ExplSyntaxOff

\tikzstyle{minsky}=[node distance=3.3cm,auto,font=\small] 

\newcommand{\numIT}{\#_{\mathit{IT}}}

\tikzstyle{trans}=[->,decorate, decoration={snake,amplitude=.3mm,segment length=2mm,post length=1mm},shorten >=1pt]
\tikzstyle{inhib}=[-o,shorten >=1pt]

\tikzstyle{query}=[rectangle,rounded corners=2.2pt,draw=black,thin] 
\tikzstyle{nonblocked}=[]
\tikzstyle{nonquery}=[rectangle,rounded corners=2.2pt,draw=black,thin]
\tikzstyle{processed}=[very thick]
\tikzstyle{pruned}=[]
\tikzstyle{pending}=[]
\tikzstyle{tononblocked}=[]

\tikzstyle{tikztransrel}=[to reversed->]

\tikzset{petri net/.style={scale=1.0,bend angle=45,every transition/.style={fill,minimum width=1mm,minimum height=3.5mm},every place/.style={draw,thick,minimum size=4mm}}}

\tikzstyle{placea}=[place,label=above:{$p_1$}]
\tikzstyle{placeb}=[place,label=above:{$p_2$}]
\tikzstyle{placec}=[place,label=above:{$p_3$}]
\tikzstyle{placed}=[place,label=above:{$p_4$}]
\tikzstyle{thread_trans}=[|->] 

\pgfplotstableset{col sep=comma,x=Row,}
\pgfplotsset{ccomp/.style= {
	footnotesize,
	xmin=1,xmax=\totalbenchmarksnum,ymin=0.1,ymax=1800,
	width=\ccompwidth,height=\ccompheight,
	minor xtick={0,1,2,...,\totalbenchmarksnum},
	xtick={0,5,...,\totalbenchmarksnum},
	xticklabels={0,5,...,\totalbenchmarksnum},
	yticklabels={,,1$s$,10$s$,100$s$,1000$s$},
	legend style={draw=none},
}}

\tikzstyle{momark}=[mark=*]
\tikzstyle{ormark}=[mark=diamond*]
\tikzstyle{rgmark}=[mark=square*]
\tikzstyle{samark}=[mark=asterisk]
\tikzstyle{btcmark}=[mark=*]
\tikzstyle{btcmarkopt}=[mark=10-pointed star]
\tikzstyle{dcdmark}=[mark=diamond*]
\tikzstyle{petmark}=[mark=square*]
\tikzstyle{mbrmark}=[mark=asterisk]

\newcommand{\HLGM}[1]{#1}
\global\long\def\ie{\HLGM{i.e.}} 
\global\long\def\eg{\HLGM{e.g.}} 
\global\long\def\wrt{\HLGM{with respect to}} 
\global\long\def\respectively{\HLGM{respectively}} 
\global\long\def\respectivelyend{\HLGM{respectively}} 

\global\long\def\wqodness{well quasi-orderedness}

\newcommand{\dr}{\HLGM{dual-ref\-er\-ence}}
\newcommand{\DR}{\HLGM{Dual-Ref\-er\-ence}}
\newcommand{\Dr}{\HLGM{Dual-ref\-er\-ence}}
\newcommand{\drabb}{\HLGM{DR}}

\newcommand{\execution}{execution}

\newcommand{\tablesizeb}[1]{\footnotesize{#1}}

\tdplotsetmaincoords{60}{134}
\tikzstyle{quader}=[scale=1.2]
\tikzstyle{plainst}=[fill, pattern=north east lines, draw=black,line width=1pt]

\global\long\def\C{\HLGM{\text{C}}}

\global\long\def\mathperiod{\HLGM{\,\text{.}}}

\global\long\def\mathcomma{\HLGM{\,\text{,}}}

\global\long\def\mathprogram#1{\HLGM{\mathtt{#1}}}
\global\long\def\equalsop{\HLGM{=}}
\global\long\def\notequalsop{\HLGM{\ne}}
\global\long\def\assignop{\HLGM{\colonequals}}
\global\long\def\plusop{\HLGM{+}}

\global\long\def\st{\HLGM{{\ : \ }}}
\global\long\def\suchthat{\HLGM{\operatorname:}}
\global\long\def\someindex{\HLGM{i}}

\global\long\def\impliesop{\HLGM{\Rightarrow}}

\global\long\def\true{\HLGM{\mathit{true}}}

\newcommand{\booleantext}[1]{\mathtt{#1}}

\global\long\def\truetext{\booleantext T}
\global\long\def\falsetext{\booleantext F}

\global\long\def\truefalsetext{\star}

\global\long\def\andop{\HLGM{\wedge}}

\global\long\def\Andop{\HLGM{\bigwedge}}
\global\long\def\orop{\HLGM{\vee}}
\global\long\def\Orop{\HLGM{\bigvee}}
\global\long\def\iffop{\HLGM{\Leftrightarrow}}

\global\long\def\negop#1{\HLGM{\neg#1}}

\global\long\def\union{\HLGM{\cup}}
\global\long\def\Union{\HLGM{\bigcup}}

\global\long\def\program{\HLGM{\mathcal{P}}}

\global\long\def\satbound{\HLGM{\mathsf{b}}}

\global\long\def\threadcount{n}

\global\long\def\sharedvars{\HLGM{S}}

\global\long\def\localvarstext{L}
\newcommandx\localvars[1][usedefault, addprefix=\global, 1=]{\HLGM{\localvarstext_{#1}}}
\newcommandx\abslocalvars[1][usedefault, addprefix=\global, 1=]{\HLGM{\absof{\localvarstext}_{#1}}}
\newcommandx\exabslocalvars[1][usedefault, addprefix=\global, 1=]{\HLGM{\exabsof{\localvarstext}_{#1}}}
\newcommandx\localvarsp[1][usedefault, addprefix=\global, 1=a]{\HLGM{\localvars[#1]'}}

\newcommand{\lock}{\mathit{lock}}

\global\long\def\abslocalvarspas{\abslocalvars[\pasindex]}

\global\long\def\programvarstext{V}
\newcommandx\programvars[1][usedefault, addprefix=\global, 1=]{\HLGM{\programvarstext_{#1}}}
\newcommandx\absprogramvars[1][usedefault, addprefix=\global, 1=]{\HLGM{\absof{\programvarstext}_{#1}}}

\newcommandx\exabsprogramvars[1][usedefault, addprefix=\global, 1=]{\HLGM{\exabsof{\programvarstext}_{#1}}}

\global\long\def\programctr{\HLGM{\mathtt{pc}}}

\global\long\def\programvar{\HLGM{\mathtt{v}}}

\global\long\def\localboolvartext{\HLGM{b}}
\global\long\def\localboolvar{\HLGM{\mathtt{\localboolvartext}}}
\global\long\def\localboolvarp{\HLGM{\localboolvar'}}
\global\long\def\localboolvarpas{\HLGM{\localboolvar_{\pasindex}}}
\global\long\def\localboolvarpasp{\HLGM{\localboolvarpas'}}

\global\long\def\localvartext{\HLGM{m}}

\global\long\def\localvar{\HLGM{\mathtt{\localvartext}}}

\global\long\def\localvarothertext{\HLGM{l}}
\global\long\def\localvarother{\HLGM{\mathtt{\localvarothertext}}}

\global\long\def\localvarotherpas{\HLGM{\localvarother_{\pasindex}}}

\global\long\def\sharedvartext{\HLGM{s}}
\global\long\def\sharedvar{\HLGM{\mathtt{\sharedvartext}}}

\global\long\def\sharedvarothertext{\HLGM{t}}
\global\long\def\sharedvarother{\HLGM{\mathtt{\sharedvarothertext}}}

\global\long\def\localstate{\HLGM{c}}

\global\long\def\somestate{\HLGM{v}}

\global\long\def\rel{\HLGM{\mathcal{R}}}
\global\long\def\absrel{\HLGM{\absof\rel}}
\global\long\def\exabsrel{\HLGM{\exabsof\rel}}

\global\long\def\covers{\HLGM{\ge}}
\global\long\def\coveredby{\HLGM{\le}}

\global\long\def\strictlycovers{\HLGM{>}}
\global\long\def\strictlycoveredby{\HLGM{<}}

\global\long\def\tool#1{\mathsf{#1}}

\global\long\def\breach{\HLGM{\tool{breach}}}

\global\long\def\slab{\HLGM{\tool{slab}}}

\global\long\def\symmpatext{symmpa}
\global\long\def\symmpa{\HLGM{\mbox{\ensuremath{\tool{\symmpatext}}}}}
\global\long\def\duet{\HLGM{\mbox{\ensuremath{\tool{duet}}}}}

\global\long\def\boom{\HLGM{\tool{boom}}}

\global\long\def\threadertext{cream}
\global\long\def\threader{\HLGM{\tool{\threadertext}}}

\global\long\def\absof#1{\HLGM{\tilde{#1}}}
\global\long\def\exabsof#1{\HLGM{\hat{#1}}}

\global\long\def\preddef{\HLGM{\mathrel{::}}}
\global\long\def\reldef{\HLGM{\coloncolon}}

\global\long\def\predicateidx{\HLGM{c}}
\global\long\def\predicatenum{\HLGM{m}}

\global\long\def\absprogram{\HLGM{\absof{\program}}}
\global\long\def\exabsprogram{\HLGM{\exabsof{\program}}}

\global\long\def\naturals{\HLGM{\mathbb{N}}}

\global\long\def\booleans{\HLGM{\mathbb{B}}}
\global\long\def\parallelinst#1#2{\HLGM{#1^{#2}}}
\global\long\def\pvarsact{\HLGM{\localvars}}
\global\long\def\abspvarsact{\HLGM{\absof\pvarsact}}

\global\long\def\pasindex{\HLGM{P}}

\global\long\def\pvarspas{\HLGM{\localvars[\pasindex]}}
\global\long\def\abspvarspas{\HLGM{\abslocalvars[\pasindex]}}

\global\long\def\initialrel{\HLGM{\mathcal{I}}}
\global\long\def\absinitialrel{\HLGM{\absof\initialrel}}
\global\long\def\exabsinitialrel{\HLGM{\exabsof\initialrel}}

\global\long\def\labelformat#1{\HLGM{#1}}

\global\long\def\somelabel#1{\HLGM{\labelformat{\ell_{#1}}}}
\global\long\def\labela{\HLGM{\labelformat{\somelabel{1}}}}
\global\long\def\labelb{\HLGM{\labelformat{\somelabel{2}}}}
\global\long\def\labelc{\HLGM{\labelformat{\somelabel{3}}}}
\global\long\def\labeld{\HLGM{\somelabel{4}}}

\global\long\def\pca{\HLGM{\labela\mathprogram :}}
\global\long\def\pcb{\HLGM{\labelb\mathprogram :}}
\global\long\def\pcy{\HLGM{\labelc\mathprogram :}}

\global\long\def\pci{\HLGM{\phantom{\labeld\mathprogram :}}}

\global\long\def\localvarpas{\HLGM{\localvar_{\pasindex}}}

\global\long\def\localvarotherpas{\HLGM{\localvarother_{\pasindex}}}

\global\long\def\labelend{\HLGM{\somelabel{e}}}

\global\long\def\relreset{\HLGM{\mathcal{F}}} 
\global\long\def\pred#1{\HLGM{\predicate[#1]}}

\global\long\def\bpredo{\HLGM{\mathtt{b}}} 
\global\long\def\bpredopas{\bpredo_{\pasindex}} 
\global\long\def\bpredop{\bpredo'} 
\global\long\def\bpredopasp{\bpredopas'} 

\global\long\def\bpred#1{\HLGM{\mathtt{b}[#1]}} 

\global\long\def\bpredpas#1{\HLGM{\bpred{#1}_{\pasindex}}}

\global\long\def\cp{\HLGM{\bpredpas 3}}

\global\long\def\ap{\HLGM{\bpredpas 1}}

\global\long\def\pc{\HLGM{\programctr}}

\global\long\def\pcp{\HLGM{\programctr_{\pasindex}}}
\global\long\def\PREDa{\HLGM{\bpred 1}}
\global\long\def\PREDb{\HLGM{\bpred 2}}
\global\long\def\PREDc{\HLGM{\bpred 3}}
\global\long\def\bp{\HLGM{\PREDb_{\pasindex}}}

\global\long\def\absprogrammon{\HLGM{\absprogram_{m}}}

\global\long\def\absrelmon{\HLGM{\absof{\rel}_{m}}}
\global\long\def\assertion{\HLGM{f}}

\global\long\def\idxact{\HLGM{a}}
\global\long\def\idxpas{\HLGM{p}}

\global\long\def\subst#1#2#3{\HLGM{{{#1}\{{#3}\!\triangleright\!{#2}\}}}}
\global\long\def\substp#1#2#3{\HLGM{{{#1}\{{#3}\!\triangleright\mkern-13.0mu\raisebox{1.5pt}{$\triangleright$}{#2}\}}}}

\global\long\def\MInc#1{\HLGM{#1{\texttt{++}}}}
\global\long\def\MDec#1{\HLGM{#1{\texttt{{-}-}}}}
\global\long\def\eqtestz{\HLGM{{\mathrel{\stackrel{\scriptscriptstyle ?}{=}}}}}
\global\long\def\MTZ#1{\HLGM{#1 \eqtestz 0}} 

\global\long\def\localstateo{\HLGM{\localstate_{1}}}
\global\long\def\localstateoo{\HLGM{\localstate_{2}}}

\global\long\def\sharedstateo{\HLGM{s_{0}}}
\global\long\def\sharedstateoo{\HLGM{s_{1}}}
\global\long\def\sharedstateooo{\HLGM{s_{2}}}
\global\long\def\sharedstateoooo{\HLGM{s_{3}}}
\global\long\def\sharedstateooooo{\HLGM{s_{4}}}
\global\long\def\captionheadperiod#1{\HLGM{.}}

\global\long\def\citenopassivesupport{\HLGM{\cite{DBLP:conf/cav/TorreMP10,DBLP:conf/tacas/DragerKFW10,FarzanK12}}}

\newcommand{\astate}[4]{\HLGM{(#4/{#1}{}{#2}{}{#3})}}

\newcommand{\astateformula}[4]{\HLGM{%
\pc=#4\andop
\ifthenelse{1=#1}{\PREDa}{\negop\PREDa}\andop
\ifthenelse{1=#2}{\PREDb}{\negop\PREDb}\andop
\ifthenelse{1=#3}{\PREDc}{\negop\PREDc}}}
\newcommand{\astateformulapas}[4]{\HLGM{%
\pcp=#4\andop
\ifthenelse{1=#1}{\ap}{\negop\ap}\andop
\ifthenelse{1=#2}{\bp}{\negop\bp}\andop
\ifthenelse{1=#3}{\cp}{\negop\cp}}}

\newcommandx\preda[2][usedefault, addprefix=\global, 1=, 2=\pasindex]{\localvar_{#1}\ne\localvar_{#2}}
\newcommandx\predaneg[2][usedefault, addprefix=\global, 1=, 2=\pasindex]{\localvar_{#1}=\localvar_{#2}}
\newcommandx\predb[1][usedefault, addprefix=\global, 1=]{\sharedvarother>\localvar_{#1}}
\newcommandx\predc[1][usedefault, addprefix=\global, 1=]{\sharedvar=\localvar_{#1}}

\newcommand{\allequal}{\stackrel{\circ}=}

\global\long\def\naturals{\mathbb{N}}

\tikzstyle{covprereledge}=[<-,latex-left hook] 
\tikzstyle{covprereledgewid}=[draw=none]
\tikzstyle{covprereledgewidviz}=[covprereledge,dashed] 

\tikzstyle{notpruned}=[text=black]
\tikzstyle{coversedge}=[densely dotted]
\tikzstyle{bck}=[opacity=0.3]\pgfplotsset{my personal style/.style= {grid=major,font=\large}}

\pgfplotsset{run time profile/.style= {every non boxed x axis/.style={xtick align=center,enlarge x limits=false,x axis line style={-|}},width=10.0cm,height=6.0cm,axis y line=middle,axis x line=middle,ytick=\empty,scaled ticks=false,}}
\pgfplotsset{run time profile back/.style= {scaled ticks=false,hide x axis,xtick=\empty,xticklabels=\empty,}}

\newcommand{\stateset}{V}

\newcommand{\transrel}{\rightarrowtail}

\newcommand{\ws}{well-quasi ordered} 

\renewcommand{\stateset}{\Sigma} 

\tikzstyle{minsky}=[node distance=4.1cm,auto,font=\small] 
\tikzstyle{reachtrees}=[yscale=0.5,rotate=0,font=\scriptsize,scale=\tikzscale]

\tikzstyle{proof construction}=[xscale=2.2,yscale=1.6,font=\small] 

\usepackage{alltt}
\usepackage{marginnote}
\usepackage{setspace}

\newcommand{\inlinedef}[1]{\emph{#1}}

\newcommand{\margin}[1]{}

\usepackage{mathtools}

\usepackage{xparse}

\RenewDocumentCommand\parallelinst{mmo}{%
  \IfNoValueTF{#3}
  {#1^{#2}}
  {#1(#3)^{#2}}%
}

\NewDocumentCommand\subsc{mmo}{%
  \IfNoValueTF{#3}%
  {{#1}_{\mathrlap{#2}}}%
  {{#1}_{\mathrlap{#2:#3}}}%
}

\NewDocumentCommand\Afunc{ooo}{%
  \mathcal{C}
  \IfNoValueTF{#2}%
  {}%
  {_{#2}}%
  \IfNoValueTF{#3}
  {}%
  {(#3)}%
  }

\NewDocumentCommand\Arel{ooo}{%
  \mathcal{D}%
  \IfNoValueTF{#1}
  {}%
  {^{#1}}%
  \IfNoValueTF{#2}
  {}%
  {_{#2}}%
  \IfNoValueTF{#3}
  {}%
  {(#3)}%
  }

\NewDocumentCommand\absfunc{oo}{%
  \alpha
  \IfNoValueTF{#2}%
  {}%
  {(#2)}%
}

\ExplSyntaxOn
\renewcommand{\LOCALSTATE}[1]{
\int_to_symbols:nnn {#1} {18202224}%
{%
{0}{\astate{\falsetext}{\falsetext}{\falsetext}{\labela}}%
{1}{\astate{\falsetext}{\falsetext}{\truetext}{\labela}}%
{2}{\astate{\falsetext}{\truetext}{\falsetext}{\labela}}%
{3}{\astate{\falsetext}{\truetext}{\truetext}{\labela}}%
{4}{\astate{\truetext}{\falsetext}{\falsetext}{\labela}}%
{5}{\astate{\truetext}{\falsetext}{\truetext}{\labela}}%
{6}{\astate{\truetext}{\truetext}{\falsetext}{\labela}}%
{7}{\astate{\truetext}{\truetext}{\truetext}{\labela}}%
{8}{\astate{\falsetext}{\falsetext}{\falsetext}{\labelb}}%
{9}{\astate{\falsetext}{\falsetext}{\truetext}{\labelb}}%
{10}{\astate{\falsetext}{\truetext}{\falsetext}{\labelb}}%
{11}{\astate{\falsetext}{\truetext}{\truetext}{\labelb}}%
{12}{\astate{\truetext}{\falsetext}{\falsetext}{\labelb}}%
{13}{\astate{\truetext}{\falsetext}{\truetext}{\labelb}}%
{14}{\astate{\truetext}{\truetext}{\falsetext}{\labelb}}%
{15}{\astate{\truetext}{\truetext}{\truetext}{\labelb}}%
{16}{\astate{\falsetext}{\falsetext}{\falsetext}{\labelc}}%
{17}{\astate{\falsetext}{\falsetext}{\truetext}{\labelc}}%
{18}{\astate{\falsetext}{\truetext}{\falsetext}{\labelc}}%
{19}{\astate{\falsetext}{\truetext}{\truetext}{\labelc}}%
{20}{\astate{\truetext}{\falsetext}{\falsetext}{\labelc}}%
{21}{\astate{\truetext}{\falsetext}{\truetext}{\labelc}}%
{22}{\astate{\truetext}{\truetext}{\falsetext}{\labelc}}%
{23}{\astate{\truetext}{\truetext}{\truetext}{\labelc}}%
{25}{\astate{\truefalsetext}{\truefalsetext}{\truefalsetext}{\labelc}}%
{26}{end}%
{146}{\labela/(\PREDa\ne\PREDc)\andop(\PREDc\impliesop\negop\PREDb)}
{172022}{\labelc/(\PREDa\ne\PREDc)\andop(\PREDc\impliesop\negop\PREDb)}
{18202223}{\labelc/(\negop\PREDc\andop(\PREDb\orop\PREDa)) \orop (\PREDa \andop \PREDb)}
{182022}{\labelc/\negop\PREDc\andop(\PREDa\orop\PREDb)}
{2022}{\labelc/\PREDa\andop\negop\PREDc}
{2456}{\labela/(\PREDa\andop\negop\PREDb)\orop(\PREDb\andop\negop\PREDc)}
}%
}
\ExplSyntaxOff

\RenewDocumentCommand\astateformula{mmmmo}{%
  \IfNoValueTF{#5}{
    \pc=#4\andop
    \ifthenelse{1=#1}{\bpred1}{\negop {\bpred1}}\andop
    \ifthenelse{1=#2}{\bpred2}{\negop {\bpred2}}\andop
    \ifthenelse{1=#3}{\bpred3}{\negop {\bpred3}}
  }{
    \pc_{#5}=#4\andop
    \ifthenelse{1=#1}{{\bpred1}_{#5}}{\negop {{\bpred1}_{#5}}}\andop
    \ifthenelse{1=#2}{{\bpred2}_{#5}}{\negop {{\bpred2}_{#5}}}\andop
    \ifthenelse{1=#3}{{\bpred3}_{#5}}{\negop {{\bpred3}_{#5}}}
  }%
}

\newcommand{\concref}{HJM04,CCKRT06,CKS07,WBKW07,%
GPR11,DonaldsonFMSD12,FarzanK12,FarzanK13}

\NewDocumentCommand\itembtooldescref{mo}{%
  \IfNoValueTF{#2}
  {\item[\textmd{\em #1:}]}
  {\item[\textmd{\em #1 (#2):}]}%
}

\lstdefinestyle{defcode}
{
tabsize=8,
breaklines=true,
language=C,
morekeywords={fork,wait,signal,broadcast,atomic,cv_broadcast,cv_signal,cv_wait,mtx_unlock,mtx_lock,cv_wait_sig,shared,local,cond},
alsoletter={_},
}

\newcommand{\refcapital}[3]{#1#3} 

\newcommand{\reffull}[2]{#1} 
\renewcommand{\reffull}[2]{#2} 

\newcommand{\appendixref}   [2][!]{\genericref[#1] A a {\reffull{ppendix}   {pp.}} {appendix}   {#2}}

\newcommand{\corollaryref}  [2][!]{\genericref[#1] C c {\reffull{orollary}  {or.}} {corollary}  {#2}}
\newcommand{\definitionref} [2][!]{\genericref[#1] D d {\reffull{efinition} {ef.}} {definition} {#2}}
\newcommand{\figureref}     [2][!]{\genericref[#1] F f {\reffull{igure}     {ig.}} {figure}     {#2}}

\newcommand{\lemmaref}      [2][!]{\genericref[#1] L l {\reffull{emma}      {emma}}{lemma}      {#2}}

\newcommand{\sectionref}    [2][!]{\genericref[#1] S s {\reffull{ection}    {ect.}}{section}    {#2}}
\newcommand{\tableref}      [2][!]{\genericref[#1] T t {\reffull{able}      {able}}{table}      {#2}}
\newcommand{\theoremref}    [2][!]{\genericref[#1] T t {\reffull{heorem}    {hm.}} {theorem}    {#2}}

\newcommand{\Appendixref}   [2][!]{\Genericref[#1]{\reffull{Appendix}   {App.}} {appendix}   {#2}}

\newcommand{\Figureref}     [2][!]{\Genericref[#1]{\reffull{Figure}     {Fig.}} {figure}     {#2}}

\newcommand{\Lemmaref}      [2][!]{\Genericref[#1]{\reffull{Lemma}      {Lemma}}{lemma}      {#2}}

\newcommand{\Tableref}      [2][!]{\Genericref[#1]{\reffull{Table}      {Table}}{table}      {#2}}
\newcommand{\Theoremref}    [2][!]{\Genericref[#1]{\reffull{Theorem}    {Thm.}} {theorem}    {#2}}

\newcommand{\equationref}[2][!]{\Ifthen{\Equal{#1}!}{\refcapital E e {\reffull{quation} {q.}}~}(\ref{equation: #2})}
\newcommand{\Equationref}[2][!]{\Ifthen{\Equal{#1}!}                 {\reffull{Equation}{Eq.}~}(\ref{equation: #2})}

\newcommand{\genericref}   [6][!]{\Ifthen{\Equal{#1}{!}}{\refcapital{#2}{#3}{#4}~}\ref{#5: #6}} 
\newcommand{\Genericref}   [4][!]{\Ifthen{\Equal{#1}{!}}                   {{#2}~}\ref{#3: #4}} 

\tikzstyle{reachtrees}=[yscale=0.7,rotate=0,font=\scriptsize,scale=\tikzscale]

\newcounter{inlinerefcounter}

\crefname{inlinerefcounter}{Hook}{Hooks}

\newcommand{\range}[3][X]{\Ifthen{\Equal{#1}{X}}{\{}#2,\ldots,#3\Ifthen{\Equal{#1}{X}}{\}}} 

\newcommand{\predsym} Q
\renewcommand{\predicateidx} i

\renewcommand{\pred}[2][]{\Ifthenelse{\Empty{#1}}{\predsym[#2]}{\predsym[#2]_{#1}}}
\newcommand{\predact}[1]{\pred[\idxact]{#1}}

\renewcommand{\bpred}[2][]{\Ifthenelse{\Empty{#1}}{b[#2]}{b[#2]_{#1}}}
\newcommand{\bpredact}[1]{\bpred[\idxact]{#1}}

\renewcommand{\bpredpas}[1]{\bpred[\pasindex]{#1}}

\renewcommand{\booleantext}[1]{\mathtt{#1}}
\renewcommand{\truetext} {\booleantext T}
\renewcommand{\falsetext}{\booleantext F}
\renewcommand{\absfunc}{\alpha}

\renewcommand{\parallelinst}[3][]{\Ifthenelse{\Empty{#1}}{{#2}^{#3}}{({#2}_{#1})^{#3}}}

\newcommand{\otherstate}    w
\newcommand{\yetotherstate} u

\newcommand{\slice}[2]{\angles{#1,#2}}

\renewcommand{\labelend}{\ell_{\bot}}

\global\long\def\covers{\succeq}
\global\long\def\coveredby{\preceq}

\global\long\def\strictlycovers{\succ}
\global\long\def\strictlycoveredby{\prec}

\global\long\def\sharedstateo{\HLGM{{0}}}
\global\long\def\sharedstateoo{\HLGM{{1}}}
\global\long\def\sharedstateooo{\HLGM{{2}}}
\global\long\def\sharedstateoooo{\HLGM{{3}}}
\global\long\def\sharedstateooooo{\HLGM{{4}}}

\global\long\def\localvartext{\HLGM{l}}
\global\long\def\localvarothertext{\HLGM{m}}
\global\long\def\localvaryetothertext{\HLGM{k}}

\global\long\def\localvaryetother{\HLGM{\mathtt{\localvaryetothertext}}}

\global\long\def\localvaryetotherpas{\HLGM{\localvaryetother_{\pasindex}}}

\newcommand{\href}[2]{\url{#1}}

\begin{document}

\author{Alexander Kaiser\inst 1 \and Daniel Kroening\inst 1 \and Thomas Wahl\inst 2}
\institute{University of Oxford, United Kingdom \and Northeastern University, Boston, United States}

\title{Lost in Abstraction: \\
  Monotonicity in Multi-Threaded Programs\\
  (Extended Technical Report)%
  \thanks{This work is supported by the Toyota Motor Corporation,
  NSF grant no.~1253331 and ERC project~280053.}}

\maketitle

\begin{abstract}
  \emph{Monotonicity} in concurrent systems stipulates that, in any global
  state, extant system actions remain executable when new processes are added to
  the state. This concept is not only natural and common in
  multi-threaded software, but also useful: if every thread's memory is
  finite, monotonicity often guarantees the decidability of safety property
  verification even when the number of running threads is unknown. In this
  paper, we show that the act of obtaining finite-data thread abstractions
  for model checking can be at odds with monotonicity:
  Predicate-abstracting certain widely used monotone software results in
  non-monotone multi-threaded Boolean programs ---
  the monotonicity is \emph{lost in the abstraction}. As a result,
  well-established sound and complete safety checking algorithms become
  inapplicable; in fact, safety checking turns out to be undecidable for
  the obtained class of unbounded-thread Boolean programs. We demonstrate
  how the abstract programs can be modified into monotone ones, without
  affecting safety properties of the non-monotone abstraction. This
  significantly improves earlier approaches of enforcing monotonicity via
  overapproximations.
\end{abstract}

\section{Introduction}
\label{section: Introduction}

This paper addresses non-recursive procedures executed by multiple threads
(\eg\ dynamically generated, and possibly unbounded in number), which
communicate via shared variables or higher-level mechanisms such as
mutexes. OS-level code, including Windows, UNIX, and Mac OS device drivers,
makes frequent use of such concurrency APIs, whose correct use is therefore
critical to ensure a reliable programming environment.

The utility of \emph{predicate abstraction} as a safety analysis method is
known to depend critically on the choice of predicates: the consequences of
a poor choice range from inferior performance to flat-out unprovability of
certain properties.
We propose in this paper an extension of predicate abstraction to
multi-threaded programs that enables
reasoning about intricate data relationships, namely
\begin{description}

\item[\textbf{shared-variable:}] ``shared variables $\sharedvar$ and
  $\sharedvarother$ are equal'',

\item[\textbf{single-thread:}] ``local variable $\localvar$ of thread
  $\someindex$ is less than shared variable $\sharedvar$'', and

\item[\textbf{inter-thread:}] ``local variable $\localvar$ of thread
  $\someindex$ is less than variable $\localvar$ \emph{in all other
  threads}''.

\end{description}

Why such a rich predicate language? For certain concurrent algorithms such
as the widely used \emph{ticket} busy-wait lock algorithm
\cite{Andrews:1991:CPP:110561} (the default locking mechanism in the Linux
kernel since 2008; see \figureref{ticket algorithm}), the verification of
elementary safety properties \emphasize{requires} single- and inter-thread
relationships. They are needed to express, for instance, that a thread
holds the minimum ticket value, an inter-thread relationship.

In the main part of the paper, we address the problem of full parameterized
(un\-bound\-ed-thread) program verification with respect to our rich
predicate language. Such reasoning requires first that the
$\threadcount$-thread abstract program
$\parallelinst{\exabsprogram}{\threadcount}$, obtained by existential
inter-thread predicate abstraction of the $\threadcount$-thread concrete
program $\parallelinst{\program}{\threadcount}$, is rewritten into a single
template program $\absprogram$ to be executed by (any number of) multiple
threads.
In order to capture the semantics of these programs in the template
$\absprogram$, the template programming language must itself permit
variables that refer to the currently executing or a generic passive
thread; we call such programs \emph{\dr\ (\drabb)}. We describe how to
obtain $\absprogram$, namely essentially as an overapproximation of
$\parallelinst{\exabsprogram}{\satbound}$, for a constant $\satbound$ that
scales linearly with the number of inter-thread predicates used in the
predicate abstraction.

Given the \emph{Boolean} \dr\ program $\absprogram$, we might now expect
the unbounded-thread replicated program
$\parallelinst{\absprogram}{\infty}$ to form a classical \emph{well
  quasi-ordered transition system}
\cite{ACJY96}, enabling the fully automated, algorithmic safety property
verification in the abstract. This turns out not to be the case: the
expressiveness of \dr\ programs renders parameterized program location
reachability undecidable, despite the finite-domain variables. The root
cause is the lack of \emph{monotonicity} of the transition relation
\wrt\ the standard partial order over the space of unbounded thread
counters. That is, adding passive threads to the source state of a valid
transition can invalidate this transition and in fact block the system.
Since the input \C\ programs are, by contrast, perfectly monotone, we say
the monotonicity is \emph{lost in the abstraction}. As a result, our
abstract programs are in fact not well quasi-ordered.

Inspired by earlier work on \emph{monotonic abstractions}
\cite{DBLP:journals/entcs/AbdullaDR08}, we address this problem by
restoring the monotonicity using a simple \emph{closure operator}, which
enriches the transition relation of the abstract program $\absprogram$
such that the obtained program $\absprogrammon$ engenders a monotone (and
thus well quasi-ordered) system.
The
closure operator essentially terminates passive threads that block
transitions allowed by other passive threads. In contrast to those earlier
approaches, which \emph{enforce} (rather than restore) monotonicity in
genuinely non-monotone systems,
we exploit the fact that the input programs are monotone. As a
result, the monotonicity closure $\absprogrammon$ can be shown to be
\emph{safety-equivalent} to
the intermediate
program $\absprogram$.

To summarize, the central contribution of this paper is a predicate
abstraction strategy for unbounded-thread
\C\ programs, with respect to the rich language of inter-thread predicates.
This language allows the abstraction to track properties that are
essentially universally quantified over all passive threads. To this end,
we first develop such a strategy for a fixed number of threads. Second, in
preparation for extending it to the unbounded case, we describe how the
abstract model, obtained by existential predicate abstraction for a given
thread count~$\threadcount$, can be expressed as a template program that
can be multiply instantiated. Third, we show a sound and complete algorithm
for reachability analysis for the obtained parameterized Boolean
\dr\ programs. We overcome the undecidability of the problem by building a
monotone closure that enjoys the same safety properties as the original
abstract \dr\ program.

We omit in this submission practical aspects such as predicate discovery,
the algorithmic construction of the abstract programs, and abstraction
refinement. We provide, however, an extensive appendix, with proofs of all lemmas and theorems.

\begin{figure}[tbp]
  \fbox{%
    \begin{minipage}[c]{.45\textwidth}
\begin{lstlisting}
struct Spinlock {
  $\pci$ natural $\sharedvartext$ $\assignop$ 1; // ticket being served
  $\pci$ natural $\sharedvarothertext$ $\assignop$ 1; }; // next free ticket

struct Spinlock $\lock$; // shared

void spin_lock() {
  $\pci$ natural $\localvartext$ $\assignop$ 0; // local
  $\pca$ $\localvartext$ $\assignop$ fetch_and_add($\lock.\sharedvarothertext$);
  $\pcb$ while ($\localvartext$ $\notequalsop$ $\lock.\sharedvartext$) 
	   $\pci$ $\pci$ /* spin */; }

void spin_unlock() {
  $\pcy$ $\lock.\sharedvartext$++; }
\end{lstlisting}%
  \end{minipage}}
  \hfill
  \begin{minipage}[c]{.4\textwidth}
    \footnotesize
    \emphasize{The ticket algorithm:} Shared variable $\lock$ has two
    integer components: $\sharedvar$~holds the ticket currently served (or,
    if none, the ticket served next), while $\sharedvarother$ holds the
    ticket to be served after all waiting threads have had access.

    To request access to the locked region, a thread atomically retrieves
    the value of $\sharedvarother$ and then increments $\sharedvarother$.
    The thread then busy-waits (``spins'') until local variable $\localvar$
    agrees with shared $\sharedvar$. To unlock, a thread increments
    $\sharedvar$.

    \

    See \appendixref{Inter-thread Predicates are Essential for the Ticket
      Algorithm} for more intuition.
  \end{minipage}
  \hfill
  \mbox{}
  \caption[The ticket algorithm]{Our goal is to verify ``unbounded-thread
    mutual exclusion'': no matter how many threads try to acquire and
    release the lock concurrently, no two of them should simultaneously be
    between the calls to functions \lstinline!spin_lock! and
    \lstinline!spin_unlock!.}
  \label{figure: ticket algorithm}
\end{figure}%

\section{Inter-Thread Predicate Abstraction}
\label{section: Inter-Thread Predicate Abstraction}

In this section we introduce single- and inter-thread predicates, with
respect to which we then formalize existential predicate abstraction.
Except for the extended predicate language, these concepts are mostly
standard and lay the technical foundations for the contributions of this
paper.

\subsection{Input Programs and Predicate Language}

\subsubsection{Asynchronous Programs}
\label{section: Asynchronous Programs}

An \emph{asynchronous program} $\program$ allows only one thread at a time
to change its local state. We model $\program$, designed for execution by
$\threadcount \atl 1$ concurrent threads, as follows. The variable
set $\programvars$ of a program $\program$ is partitioned into sets
$\sharedvars$ and $\pvarsact$. The variables in $\sharedvars$, called
\inlinedef{shared}, are accessible jointly by all threads, and those in
$\localvars$, called \inlinedef{local}, are accessible by the individual
thread that owns the variable. We assume the statements of $\program$ are
given by a {transition formula} $\rel$ over unprimed (current-state) and
primed (next-state) variables, $\programvars$ and $\programvars' =
\{\programvar'\st \programvar \in \programvars\}$. Further, the initial
states are characterized by the initial
formula $\initialrel$
over $\programvars$. We assume $\initialrel$ is expressible in a suitable
logic for which existential quantification is computable (required later
for the abstraction step).

As usual, the computation may be controlled by a local program counter
$\programctr$, and involve non-recursive function calls. When executed by
$\threadcount$ threads, $\program$ gives rise to
\inlinedef{$\threadcount$-thread program states} consisting of the
valuations of the variables in
$\programvars[\threadcount] = \sharedvars \union \localvars[1] \union \ldots \localvars[\threadcount]$,
where $\localvars[\someindex] = \{\localvar_\someindex \st \localvar \in \localvars\}$.
We call a variable set \inlinedef{uniformly indexed} if its variables either
all have no index, or all have the same index. For a formula $\assertion$
and two uniformly-indexed variable sets $X_1$ and $X_2$,
let $\subst f {X_2}{X_1}$ denote $\assertion$ after replacing every
occurrence of a variable in $X_1$ by the variable in $X_2$ with the same
base name, if any; unreplaced if none. We write $\substp f{X_2}{X_1}$ short
for $\subst{\subst f{X_2}{X_1}}{{X_2}'}{{X_1}'}$. As an example, given
$\sharedvars = \{\sharedvar\}$ and $\localvars = \{\localvar\}$, we have
$\substp{(\localvar' = \localvar +
  \sharedvar)}{\localvars[\idxact]}{\localvars} = (\localvar_\idxact' =
\localvar_\idxact + \sharedvar)$. Finally, let $X \allequal X'$ stand for
$\forall x \in X \suchthat x \equalsop x'$.

The \inlinedef{$\threadcount$-thread instantiation}
$\parallelinst{\program}{\threadcount}$ is defined for $\threadcount \atl
1$ as
\begin{equation}
  \parallelinst{\program}{\threadcount} = (\parallelinst{\rel}{\threadcount},\parallelinst{\initialrel}{\threadcount}) =
  \left(
       {\Orop}_{\!\!\idxact=1}^\threadcount \parallelinst[\idxact]{\rel}{\threadcount}, \
       {\Andop}_{\idxact=1}^\threadcount \subst{\initialrel}{\localvars[\idxact]}{\pvarsact}
  \right)
  \label{equation: n thread instantiation}
\end{equation}
where
\format\vspace{-1.5\baselineskip}
\begin{align}
  \parallelinst[\idxact]{\rel}{\threadcount} & \reldef
  \substp{\rel}{\localvars[\idxact]}{\pvarsact}
  \andop
  \subsc{\Andop}{\idxpas}[\idxpas\ne\idxact] \ \localvars[\idxpas] \allequal \localvarsp[\idxpas] \mathperiod
  \label{equation: mtp_exp_rel}
\end{align}
\noindent
Formula $\parallelinst[\idxact]{\rel}{\threadcount}$ asserts that the
shared variables, and the variables of the \emph{active} (executing) thread
$\idxact$ are updated according to $\rel$, while the local variables of
passive threads $\idxpas\ne\idxact$ are not modified ($\idxpas$ ranges over
$\{1,\ldots,n\}$). A state is \inlinedef{initial} if all threads are in a
state satisfying~$\initialrel$. An \inlinedef{$\threadcount$-thread
  \execution} is a sequence of $\threadcount$-thread program states whose
first state satisfies $\parallelinst{\initialrel}{\threadcount}$ and whose
consecutive states are related by $\parallelinst{\rel}{\threadcount}$. We
assume the existence of an error location in $\program$; an \emph{error
  state} is one where some thread resides in the error location. $\program$
is \emph{safe} if no execution exists that ends in an error state. Mutex
conditions can be checked using a ghost semaphore and redirecting threads
to the error location if they try to access the critical section while the
semaphore is set.

\subsubsection{Predicate Language}

We extend the predicate language from \cite{DonaldsonFMSD12} to allow the
use of the \emph{passive-thread variables} $\pvarspas =
\{\localvar_{\pasindex} \st \localvar \in \localvars\}$, each of which
represents a local variable owned by a generic passive thread. The presence
of variables of various categories gives rise to the following predicate
classification.
\begin{DEF}
  A predicate $\predsym$ over $\sharedvars$, $\localvars$ and $\pvarspas$
  is \emphdef{shared} if it contains variables from $\sharedvars$ only,
  \emphdef{local} if it contains variables from $\localvars$ only,
  \emphdef{single-thread} if it contains variables from $\localvars$ but
  not from $\pvarspas$, and \emphdef{inter-thread} if it contains variables
  from $\localvars$ and from $\pvarspas$.
  \label{definition: predicate classification}
\end{DEF}
Single- and inter-thread prediactes may contain variables from
$\sharedvars$. For example, in the ticket algorithm (\figureref{ticket
  algorithm}), with
$\sharedvars=\{\sharedvar,\sharedvarother\}$ and
$\localvars=\{\localvar\}$, examples of shared, local, single- and
inter-thread predicates are: $\sharedvar = \sharedvarother$,
$\localvar\equalsop5$, $\sharedvar=\localvar$ and $\preda$,
\respectivelyend.

\paragraph{Semantics} Let $\range[]{\pred1}{\pred{\predicatenum}}$ be
$\predicatenum$ predicates (any class). Predicate $\pred{\predicateidx}$ is
evaluated in a given $\threadcount$-thread state~$\somestate$
($\threadcount \atl 2$) with respect to a choice of active
thread $\idxact$:
\begin{equation}
  \predact{\predicateidx} \wbox{$\reldef$} \subsc{\Andop}{\idxpas}[\idxpas \not= \idxact] \ \subst{\subst{\pred{\predicateidx}}{\localvars[\idxact]}{\localvars}}{\localvars[\idxpas]}{\pvarspas} \mathperiod
  \label{equation: inter-thread semantics}
\end{equation}
As special cases, for single-thread and shared predicates (no $\pvarspas$
variables), we have $\predact i = \subst{\pred
  i}{\localvars[\idxact]}{\localvars}$ and $\predact i = \pred i$, resp. We
write $v \models \predact i$ if $\predact i$ holds in state $v$. Predicates
$\pred i$ give rise to an abstraction function~$\absfunc$, mapping each
$\threadcount$-thread program state $v$ to an $\predicatenum \times
\threadcount$ bit matrix with entries
\begin{equation}
  \absfunc(v)_{i,a} \ = \
  \begin{dcases*}
    \truetext  & if $v \models \predact i$ \\
    \falsetext & otherwise\mathperiod
  \end{dcases*}
  \label{equation: abstraction function}
\end{equation}
Function $\absfunc$ partitions the $\threadcount$-thread program state
space via $\predicatenum$ predicates into $2^{\predicatenum \times
  \threadcount}$ equivalence classes. As an example, consider the
inter-thread predicates $\localvar\le\localvarpas$,
$\localvar>\localvarpas$, and $\localvar\ne\localvarpas$ for a local
variable $\localvar$, $\threadcount = 4$ and the state $\somestate \reldef
(\localvar_1,\localvar_2,\localvar_3,\localvar_4) = (4,4,5,6)$:
\begin{equation}
  \absfunc(\somestate) =
  \begin{pmatrix}
    \truetext  & \truetext  & \falsetext & \falsetext \\
    \falsetext & \falsetext & \falsetext & \truetext  \\
    \falsetext & \falsetext & \truetext  & \truetext
  \end{pmatrix}
  \mathperiod
  \label{equation: absfunc example}
\end{equation}
In the matrix, row $\predicateidx \in \{1,2,3\}$ lists the truth of
predicate $\pred{\predicateidx}$ for each of the four threads in the active
role. Predicate $\localvar\le\localvarpas$ captures whether a thread owns
the minimum value for local variable $\localvar$ (true for $\idxact=1,2$);
$\localvar>\localvarpas$ tracks whether a thread owns the \emph{unique}
maximum value (true for $\idxact=4$) ; finally $\localvar\ne\localvarpas$
captures the uniqueness of a thread's copy of $\localvar$ (true for
$\idxact=3,4$).

\paragraph{Inter-thread predicates and abstraction} Predicates that reason
universally about threads have
been used successfully as targets in (inductive) invariant generation
procedures \cite{APRXZ01,SSSC12}. In this paper we discuss their role in
abstractions. The use of these fairly expressive and presumably expensive
predicates is not by chance: automated methods that cannot reason about
them \cite{FarzanK13CAV,DonaldsonFMSD12,WBKW07} essentially fail for the
ticket algorithm in \figureref{ticket algorithm}: for a fixed number of
threads that concurrently and repeatedly (\eg\ in an infinite loop) request
and release lock ownership, the inter-thread relationships need to be
``simulated'' via enumeration, incurring very high time and space
requirements, even for a handful of threads. In the unbounded-thread case,
they diverge. This is essentially due to known limits of thread-modular and
Owicki-Gries style proof systems, which do not have access to inter-thread
predicates \cite{malkis10}. \Appendixref{Inter-thread Predicates are
  Essential for the Ticket Algorithm} shows that the number of
\emph{single-thread} predicates needed to prove correctness of the ticket
algorithm depends on $\threadcount$,
from which unprovability in the unbounded case follows.

\subsection{Existential Inter-Thread Predicate Abstraction}
\label{section: Existential Inter-Thread Predicate Abstraction}

Embedded into our formalism, the goal of \emph{existential predicate
  abstraction} \cite{ClarkeGrumbergLong94,GrafSaidi97} is to derive an
abstract program $\parallelinst{\exabsprogram}{\threadcount}$ by treating
the equivalence classes induced by \equationref{abstraction function} as
abstract states. $\parallelinst{\exabsprogram}{\threadcount}$ thus has
$\predicatenum \times \threadcount$ Boolean variables:
\[
  \textstyle
  \exabsprogramvars[\threadcount]
  \ = \ \Union_{\idxact=1}^{\threadcount} \exabslocalvars[a]
  \ = \ \Union_{\idxact=1}^{\threadcount} \{\bpredact{\predicateidx} \suchthat 1 \atm \predicateidx \atm m\} \mathperiod
\]
Variable $\bpredact{\predicateidx}$ tracks the truth of predicate
$\pred{\predicateidx}$ for active thread~$\idxact$. This is formalized in
\equationref[]{existential predicate abstraction definition}, relating
concrete and abstract $\threadcount$-thread states (valuations of
$\programvars[\threadcount]$ and $\exabsprogramvars[\threadcount]$, resp.):
\begin{equation}
  \Arel[\threadcount] \reldef \Andop_{\predicateidx=1}^{\predicatenum} \Andop_{\idxact=1}^{\threadcount} \bpredact{\predicateidx} \iffop \predact{\predicateidx} \mathperiod
  \label{equation: existential predicate abstraction definition}
\end{equation}
For a formula $\assertion$, let $\assertion'$ denote $\assertion$ after
replacing each variable by its primed version. We then have
\begin{math}
  \parallelinst{\exabsprogram}   {\threadcount} = (
  \parallelinst{\exabsrel}       {\threadcount},
  \parallelinst{\exabsinitialrel}{\threadcount}) =
  \left(
       {\Orop}_{\!\!\idxact=1}^\threadcount \parallelinst[\idxact]{\exabsrel}{\threadcount}, \ \parallelinst{\exabsinitialrel}{\threadcount}
  \right)
\end{math}
where
\begin{eqnarray}
  \parallelinst[\idxact]{\exabsrel}       {\threadcount} & \reldef & \exists                                        \programvars[\threadcount] \programvars[\threadcount]'   \suchthat \ \parallelinst[\idxact]{\rel}       {\threadcount} \andop \Arel[\threadcount] \andop (\Arel[\threadcount])', \label{equation: existential transrel} \\[1mm]\format
  \parallelinst         {\exabsinitialrel}{\threadcount} & \reldef & \exists \pos[l]{$\programvars[\threadcount]$}{$\programvars[\threadcount] \programvars[\threadcount]'$} \suchthat \ \parallelinst         {\initialrel}{\threadcount} \andop \Arel[\threadcount]                                \label{equation: existential initialrel}
  \mathperiod
\end{eqnarray}

\newcommandx\predtmp[4][usedefault, addprefix=\global, 1=, 2=\pasindex,3=,4=<]{\localvar_{#1}{#3}#4\localvar_{#2}{#3}}

As an example, consider the decrement operation $\localvar \assignop
\localvar-1$ on a local integer variable~$\localvar$, and the inter-thread
predicate $\predtmp$. Using \equationref{existential transrel} with
$\threadcount=2$, $\idxact=1$, we get 4 abstract transitions, which are
listed in \tableref{dec_exabs}. The table shows that the abstraction is no
longer asynchronous (treating $\bpredo_1$ as belonging to thread 1,
$\bpredo_2$ to thread 2): in the highlighted transition, the executing
thread~1 changes (its pc and hence) its local state, and so does thread~2.
By contrast, on the right we have $\localvar_2 = \localvar_2'$ in all rows.
The loss of asynchrony will become relevant in \sectionref{From Existential
  to Parametric Abstraction}, where we define a suitable abstract Boolean
programming language (which then necessarily must accommodate
non-asynchronous programs).
\renewcommand{\localCommand}[1]{\cellcolor{black!15}{#1}}
\begin{table}[htbp]
  \centering
  \begin{minipage}[t]{1.0\columnwidth}
    \tablesizeb{%
      \begin{tabularx}{1\textwidth}{>{\centering}X >{\centering}X c>{\centering}X >{\centering}X >{\centering}X >{\centering}X >{\centering}X c>{\centering}X >{\centering\arraybackslash}X}
        \toprule
        $\bpredo_1$ & $\bpredo_2$ & \hphantom{aa} & $\bpredo_1'$ & $\bpredo_2'$
        & &
        $\localvar_1$ & $\localvar_2$ & \hphantom{aa} & $\localvar_1'$ & $\localvar_2'$ \\
        \cmidrule{1-5} \cmidrule{7-11}
                      $\falsetext$  &               $\falsetext$ & &               $\truetext$   &               $\falsetext$  & & 1 & 1 &  & 0 & 1 \\
        \localCommand{$\falsetext$} & \localCommand{$\truetext$} & & \localCommand{$\falsetext$} & \localCommand{$\falsetext$} & & 1 & 0 &  & 0 & 0 \\
                      $\falsetext$  &               $\truetext$  & &               $\falsetext$  &               $\truetext$   & & 2 & 0 &  & 1 & 0 \\
                      $\truetext$   &               $\falsetext$ & &               $\truetext$   &               $\falsetext$  & & 1 & 2 &  & 0 & 2 \\
        \bottomrule
      \end{tabularx}}
  \end{minipage}
  \caption{Abstraction $\parallelinst[1]{\exabsrel}{2}$ for stmt.\
    $\localvar \assignop \localvar-1$ against predicate $\predtmp$ (left);
    concrete witness transitions, i.e.\ elements of
    $\parallelinst[1]{\rel}{2}$ (right). The highlighted row indicates
    asynchrony violations}
  \label{table: dec_exabs}
\end{table}

\noindent
\emph{Proving the ticket algorithm (fixed-thread case)} \
As in any existential abstraction, the abstract program
$\parallelinst{\exabsprogram}{\threadcount}$ overapproximates (the set of
executions of) the concrete program
$\parallelinst{\program}{\threadcount}$; the former can therefore be used
to verify safety of the latter.
We illustrate this using the ticket algorithm (\figureref{ticket
  algorithm}). Consider the predicates $\pred1 \preddef
\localvar\ne\localvarpas$, $\pred2 \preddef \sharedvarother >
\max(\localvar,\localvarpas)$, and $\pred3 \preddef \sharedvar=\localvar$.
The first two are inter-thread; the third is single-thread. The predicates
assert the uniqueness of a ticket
($\pred1$),
that the next free ticket is larger than all tickets currently owned by
threads
($\pred2$),
and that a thread's ticket is currently being served
($\pred3$).
The abstract reachability tree for
$\parallelinst{\exabsprogram}{\threadcount}$ and these predicates reveals
that mutual exclusion is satisfied: there is no state with both threads
in location $\labelc$. The tree grows exponentially with
$\threadcount$.

\section{From Existential to Parametric Abstraction}
\label{section: From Existential to Parametric Abstraction}

Classical existential abstraction as described in \sectionref{Existential
  Inter-Thread Predicate Abstraction} obliterates the symmetry pres\-ent in
the concrete concurrent program, which is given as the
$\threadcount$-thread instantiation of a single-thread template $\program$:
the abstraction is instead formulated via predicates over the explicitly
expanded $\threadcount$-thread program $\parallelinst{\rel}{\threadcount}$.
As observed in previous work \cite{DonaldsonFMSD12}, such a
``symmetry-oblivious'' approach suffers from poor scalability for
fixed-thread verification problems. Moreover, \emph{parametric} reasoning
over an unknown number of threads is impossible since the abstraction
\equationref[]{existential transrel} directly depends on~$\threadcount$.

To overcome these problems, we now derive an overapproximation of
$\parallelinst{\exabsprogram}{\threadcount}$ via a generic program template
$\absprogram$ that can be instantiated for any~$\threadcount$. There is,
however, one obstacle: instantiating a program (such as $\program$)
formulated over shared variables and one copy of the thread-local variables
naturally gives rise to asynchronous concurrency. The programs resulting
from inter-thread predicate abstraction are, however, not asynchronous, as
we have seen. As a result, we need a more powerful abstract programming
language.

\subsection{\DR\ Programs}
\label{section: Dual-Reference Programs}

In contrast to asynchronous programs, the variable set $\absprogramvars$ of
a \emph{\dr\ (\drabb)} program $\absprogram$ is partitioned into two sets:
$\abspvarsact$, the local variables of the active thread as before, and
$\abspvarspas = \{\localvar_\pasindex\st \localvar \in \abslocalvars\}$.
The latter set contains passive-thread variables, which, intuitively,
regulate the behavior of non-executing threads. To simplify reasoning about
\drabb\ programs, we exclude classical shared variables from the
description: they can be simulated using the active and passive flavors of
local variables (see \appendixref{Simulating Shared Via Local Variables}).

The statements of $\absprogram$ are given by a transition formula $\absrel$
over $\absprogramvars$ and $\absprogramvars'$, now potentially including
passive-thread variables. Similarly, $\absinitialrel$ may contain variables
from~$\abspvarspas$. The $\threadcount$-thread instantiation
$\parallelinst{\absprogram}{\threadcount}$ of a \drabb\ program
$\absprogram$ is defined for $\threadcount \atl 2$ as
\begin{equation}
  \parallelinst{\absprogram}{\threadcount} = (\parallelinst{\absrel}{\threadcount},\parallelinst{\absinitialrel}{\threadcount}) =
  \left(
       {\Orop}_{\!\!\idxact=1}^\threadcount \parallelinst[\idxact]{\absrel}       {\threadcount}, \
       {\Orop}_{\!\!\idxact=1}^\threadcount \parallelinst[\idxact]{\absinitialrel}{\threadcount}
  \right)
  \label{equation: n thread DR instantiation}
\end{equation}
where
\format\vspace{-1.5\baselineskip}
\begin{align}
  \parallelinst[\idxact]{\absrel}{\threadcount} & \reldef
  \subsc{\Andop}{\idxpas}[\idxpas\ne\idxact]
  \substp{\substp{\absrel}{\abslocalvars[\idxact]}{\abspvarsact}}{\abslocalvars[\idxpas]}{\abspvarspas}
  \label{equation: DR transitions}
  \\
  \parallelinst[\idxact]{\absinitialrel}{\threadcount} & \reldef
  \subsc{\Andop}{\idxpas}[\idxpas\ne\idxact]
  \subst{\subst{\absinitialrel}{\abslocalvars[\idxact]}{\abspvarsact}}{\abslocalvars[\idxpas]}{\abspvarspas}
  \label{equation: DR initial}
\end{align}
Recall that $\substp f{X_2}{X_1}$ denotes index replacement of both
current-state and next-state variables. \Equationref{DR transitions}
encodes the effect of a transition on the active thread~$\idxact$, and
$\threadcount-1$ passive threads~$\idxpas$. The conjunction ensures that
the transition formula $\absrel$ holds no matter which thread
$\idxpas\ne\idxact$ takes the role of the passive thread: transitions that
``work'' only for select passive threads are rejected.

\subsection{Computing an Abstract \DR\ Template}

From the existential abstraction
$\parallelinst{\exabsprogram}{\threadcount}$ we derive a Boolean
\dr\ template program $\absprogram$ such that, for all $\threadcount$, the
$\threadcount$-fold instantiation
$\parallelinst{\absprogram}{\threadcount}$ overapproximates
$\parallelinst{\exabsprogram}{\threadcount}$. The variables of
$\absprogram$ are $\abslocalvars = \{\bpred\predicateidx \st 1 \le
\predicateidx \le \predicatenum\}$ and $\abspvarspas =
\{\bpredpas{\predicateidx} \st 1 \le \predicateidx \le \predicatenum\}$.
Intuitively, the transitions of $\absprogram$ are those that are feasible,
for \emphasize{some} $\threadcount$, in
$\parallelinst{\exabsprogram}{\threadcount}$, given active thread 1 and
passive thread 2. We first compute the set $\absrel(\threadcount)$ of these
transitions for fixed $\threadcount$. Formally, the components of
$\absprogram(\threadcount) =
(\absrel(\threadcount),\absinitialrel(\threadcount))$ are, for
$\threadcount \atl 2$,
\begin{align}
  \absrel(\threadcount) & \reldef
  \exists \exabslocalvars[3],\exabslocalvars[3]',\ldots,\exabslocalvars[\threadcount],\exabslocalvars[\threadcount]' \suchthat \
  \substp{\substp{\parallelinst[1]{\exabsrel}{\threadcount}}{\abslocalvars}{\exabslocalvars[1]}}{\abspvarspas}{\exabslocalvars[2]}
  \label{equation: rpa}
  \\
  \absinitialrel(\threadcount) & \reldef
  \exists \pos[l]{$\exabslocalvars[3],\ldots,\exabslocalvars[\threadcount] \suchthat$}{$\exabslocalvars[3],\exabslocalvars[3]',\ldots,\exabslocalvars[\threadcount],\exabslocalvars[\threadcount]' \suchthat$} \ \
  \subst{\subst{\pos{$\parallelinst{\exabsinitialrel}{\threadcount}$}{$\parallelinst[1]{\exabsrel}{\threadcount}$}}{\abslocalvars}{\exabslocalvars[1]}}{\abspvarspas}{\exabslocalvars[2]}
  \label{equation: ipa}
\end{align}

We apply this strategy to the earlier example of the decrement statement
$\localvar \assignop \localvar-1$. To compute \equationref{rpa} first with
$\threadcount = 2$, we need $\parallelinst[1]{\exabsrel}{2}$, which was
enumerated previously in \tableref{dec_exabs}. Simplification results in a
Boolean \drabb\ program with variables $\localboolvar$ and
$\localboolvarpas$ and transition relation
\begin{equation}
\absrel(2) =
(
\negop\localboolvar\andop
      \localboolvarpas\andop
\negop\localboolvarp
)
\orop
(
\negop\localboolvarpas\andop
      \localboolvarp\andop
\negop\localboolvarpasp
)
\label{equation: mdec-dr-n2-abs}
\mathperiod
\end{equation}
Using \equationref[]{mdec-dr-n2-abs} as the template $\absrel$ in
\equationref[]{DR transitions} generates existential abstractions of many
concrete decrement transitions; for instance, for $\threadcount=2$ and
$\idxact=1$ we get back the transition relation in \tableref{dec_exabs}.
The question is now: does \equationref[]{mdec-dr-n2-abs} suffice as a
template, i.e.\ does $\parallelinst{(\absrel(2))}{\threadcount}$
overapproximate $\parallelinst{\exabsrel}{\threadcount}$ for all
$\threadcount$? The answer is no: the abstract 3-thread transitions shown
in \tableref{dec_exabs_three} are not permitted by
$\parallelinst{(\absrel(2))}{\threadcount}$ for any $\threadcount$, since
neither $\negop\localboolvar \andop \localboolvarpas$ nor $\localboolvarp
\andop \negop\localboolvarpasp$ are satisfied for all choices of passive
threads (violations highlighted in the table).

We thus increase $\threadcount$ to 3, recompute \equationref{rpa}, and
obtain
\begin{equation}
  \absrel(3) \reldef \absrel(2) \orop (\negop\bpredo\andop\negop\bpredopas\andop\negop\bpredop\andop\negop\bpredopasp)\mathperiod
  \label{equation: mdec-dr-n3-abs}
\end{equation}
The new disjunct accommodates the abstract transitions highlighted in
\tableref{dec_exabs_three}, which were missing before.
\begin{table}[htbp]
  \centering
  \begin{minipage}[t]{1.0\columnwidth}
    \tablesizeb{%
      \begin{tabularx}{1\textwidth}{>{\centering}X >{\centering}X >{\centering}X c>{\centering}X >{\centering}X >{\centering}X >{\centering}X >{\centering}X >{\centering}X >{\centering}X c>{\centering}X >{\centering}X >{\centering\arraybackslash}X}
        \toprule
        $\bpredo_1$ & $\bpredo_2$ & $\bpredo_3$ & \hphantom{a} & $\bpredo_1'$ & $\bpredo_2'$ & $\bpredo_3'$ &  & $\localvar_1$ & $\localvar_2$ & $\localvar_3$ & \hphantom{a} & $\localvar_1'$ & $\localvar_2'$ & $\localvar_3'$ \\
        \cmidrule{1-7} \cmidrule{9-15}
        \cellcolor{black!15}{$\falsetext$} & \cellcolor{black!15}{$\falsetext$} & $\falsetext$ &  & \cellcolor{black!15}{$\falsetext$} & \cellcolor{black!15}{$\falsetext$} & $\falsetext$ & & 1 & 0 & 0 &  & 0 & 0 & 0 \\
        \cellcolor{black!15}{$\falsetext$} & \cellcolor{black!15}{$\falsetext$} & $\truetext$  &  & \cellcolor{black!15}{$\falsetext$} & \cellcolor{black!15}{$\falsetext$} & $\falsetext$ & & 1 & 1 & 0 &  & 0 & 1 & 0 \\
        \cellcolor{black!15}{$\falsetext$} & \cellcolor{black!15}{$\falsetext$} & $\truetext$  &  & \cellcolor{black!15}{$\falsetext$} & \cellcolor{black!15}{$\falsetext$} & $\truetext$  & & 2 & 1 & 0 &  & 1 & 1 & 0 \\
        \bottomrule
      \end{tabularx}%
    }
  \end{minipage}
  \caption{Part of the abstraction $\parallelinst[1]{\exabsrel}{3}$ for
    stmt.\ $\localvar \assignop \localvar-1$ against predicate $\predtmp$
    (left); concrete witness transitions (right). The highlighted elements
    are inconsistent with \equationref[]{mdec-dr-n2-abs} as a template}
  \label{table: dec_exabs_three}
\end{table}

Does $\parallelinst{(\absrel(3))}{\threadcount}$ overapproximate
$\parallelinst{\exabsrel}{\threadcount}$ for all $\threadcount$? When does
the process of increasing $\threadcount$ stop? To answer these questions,
we first state the following diagonalization lemma, which helps us prove
the overapproximation property for the template program.
\newcounterset{lemmaMPAisConservative}{\theDEF}
\newcommand{\lemmaMPAisConservative}[1][]{%
  $\parallelinst{(\absprogram(\threadcount))}{\threadcount}$
  overapproximates $\parallelinst{\exabsprogram}{\threadcount}$: For every
  $\threadcount \atl 2$ and every $\idxact$,
  $\parallelinst[\idxact]{\exabsrel}       {\threadcount} \impliesop \parallelinst[\idxact]{\absrel       (\threadcount)}{\threadcount}$ and
  $\parallelinst         {\exabsinitialrel}{\threadcount} \impliesop \parallelinst[\idxact]{\absinitialrel(\threadcount)}{\threadcount}$.%
}
\begin{LEM}
  \lemmaMPAisConservative
  \label{lemma: mpa is conservative}
\end{LEM}

We finally give a
saturation bound for the sequence $(\absprogram(\threadcount))$. Along with
the diagonalization lemma, this allows us to obtain a template program
$\absprogram$ independent of $\threadcount$, and enable parametric
reasoning in the abstract.
\newcounterset{theoremSatbound}{\theDEF}
\newcommand{\theoremSatbound}[1][]{%
  Let $\numIT$ be the number of inter-thread predicates among the
  $\pred{\predicateidx}$. Then the sequence $(\absprogram(\threadcount))$
  stabilizes at $\satbound = 4\times\numIT+2$, i.e.\ for $\threadcount \atl
  \satbound$, $\absprogram(\threadcount) = \absprogram(\satbound)$.}
\begin{THE}
  \theoremSatbound
  \label{theorem: satbound}
\end{THE}
\begin{COR}[from L.~{\lemmaref[]{mpa is conservative}},T.~{\theoremref[]{satbound}}]
  Let $\absprogram := \absprogram(\satbound)$, for $\satbound$ as in
  \theoremref{satbound}. The components of $\absprogram$ are thus
  $(\absrel,\absinitialrel) =
  (\absrel(\satbound),\absinitialrel(\satbound))$. Then, for $\threadcount
  \atl 2$, $\absprogram^n$ overapproximates
  $\parallelinst{\exabsprogram}{\threadcount}$.
  \label{corollary: saturation}
\end{COR}

Building a template \drabb\ program thus requires instantiating the
existentially abstracted transition relation for a number $\satbound$ of
threads that is linear in the number of inter-thread predicates with
respect to which to abstraction is built.

As a consequence of losing asynchrony in the abstraction, many existing model
checkers for concurrent software become inapplicable%
~\citenopassivesupport.
For a fixed thread count $\threadcount$, the problem can be circumvented by
forgoing the replicated nature of the concurrent programs, as done in
\cite{DonaldsonFMSD12} for $\boom$ tool: it proves the ticket algorithm
correct up to $\threadcount=3$, but takes a disappointing 30 minutes. The
goal of the following section is to design an efficient and, more
importantly, fully parametric solution.

\section{Unbounded-Thread Dual-Reference Programs}
\label{section: Unbounded-Thread Dual-Reference Programs}

The multi-threaded Boolean dual-reference programs
$\parallelinst{\absprogram}{\threadcount}$ resulting from
pred\-i\-cate-ab\-stract\-ing asynchronous programs against inter-thread
predicates
are
symmetric and free of recursion.
The symmetry can be exploited using classical methods that ``counterize''
the state space \cite{GS92}: a global state is encoded as a vector of
local-state counters, each of which records the number of threads currently
occupying a particular local state.

These methods are applicable to
unbounded thread numbers as well, in which case the local state counters
range over unbounded natural numbers $[0,\infty[$.
The fact that the abstract program executed by each thread is finite-state
now might suggest that the resulting infinite-state counter systems can be
modeled as vector addition systems (as done in \cite{GS92}) or, more
generally, as
\emph{well quasi-ordered transition systems} \cite{FS01,Abdulla10} (defined
below). This would give rise to sound and complete
algorithms for local-state reachability in such programs.

This strategy turns out to be wrong: the full class of Boolean
\drabb\ programs is expressive enough to render safety
checking for an unbounded number of threads undecidable, despite the
finite-domain variables: \newcounterset{theoremUndecidability}{\theDEF}
\newcommand{\theoremUndecidability}[1][]{%
  Program location reachability for Boolean \drabb\ programs run by an
  unbounded number of threads is undecidable.}
\begin{THE}
  \theoremUndecidability
  \label{theorem: undecidability}
\end{THE}
The proof reduces the halting problem for 2-counter machines to a
reachability problem for a \drabb\ program~$\absprogram$. Counter values
$c_i$ are reduced to numbers of threads in program locations $d_i$
of~$\absprogram$. A zero-test for counter $c_i$ is reduced to testing the
\emph{absence of any thread} in location $d_i$. This condition can be
expressed using passive-thread variables, but not using traditional
single-thread local variables. (Details of the proof in \appendixref{Proof
  of theorem: undecidability}.)

\Theoremref{undecidability} implies that the unbounded-counter systems
obtained from asynchronous programs are in fact \emph{not} well
quasi-ordered. How come? Can this problem be fixed, in order to permit a
complete verification method? If so, at what cost?

\subsection{Monotonicity in Dual-Reference Programs}
\label{section: Monotonicity in Dual-Reference Programs}

For a transition system $(\stateset,\transrel)$ to be \ws, we need two
conditions to be in place \cite{FS01,Abdulla10,ACJY96}:
\begin{description}

\item[{\bfseries \wqodness:}] there exists a reflexive and transitive
  binary relation $\coveredby$ on $\stateset$ such that for every infinite
  sequence $\somestate,\otherstate,\ldots$ of states in $\stateset$ there
  exist $i,j$ with $i<j$ and $\somestate_i \coveredby \somestate_j$.

\item[{\bfseries monotonicity:}] for any
  $\somestate,\somestate',\otherstate$ with
  $\somestate\transrel\somestate'$ and $\somestate \coveredby \otherstate$
  there exists $\otherstate'$ such that $\otherstate\transrel\otherstate'$
  and $\somestate' \coveredby \otherstate'$.

\end{description}
We apply this definition to the case of dual-reference programs.
Representing global states of the abstract system
$\parallelinst{\absprogram}{\threadcount}$ defined in \sectionref{From
  Existential to Parametric Abstraction} as counter tuples, we can define
$\coveredby$ as
\[
  (\range[]{n_1}{n_k}) \coveredby (\range[]{n'_1}{n'_k}) \ \reldef \ \forall i=1..k: n_i \atm n'_i
\]
where $k$ is the number of thread-local states.
We can now characterize monotonicity of \drabb\ programs as follows:
\newcounterset{lemmaMonotoneDRPrograms}{\theDEF}
\newcommand{\lemmaMonotoneDRPrograms}[1][]{%
  Let $\absrel$ be the transition relation of a DR program. Then the
  infinite-state transition system $\union_{\threadcount=1}^\infty
  \parallelinst{\absrel}{\threadcount}$ is monotone (with respect to
  $\coveredby$) exactly if, for all $k \atl 2$:
  \begin{equation}
    \label{equation: monotone absrel-n#1}
    (\somestate,\somestate') \in \parallelinst{\absrel} k \wbox{$\impliesop$} \forall l_{k+1} \ \exists l_{k+1}', \pi \suchthat \left(\slice{\somestate}{l_{k+1}},\pi(\slice{\somestate'}{l_{k+1}'})\right) \in \parallelinst{\absrel}{k+1} \ .
  \end{equation}
}
\begin{LEM}
  \lemmaMonotoneDRPrograms
  \label{lemma: monotone DR programs}
\end{LEM}
\noindent
In \equationref[]{monotone absrel-n}, the expression $\forall l_{k+1}
\exists l_{k+1}' \ldots$ quantifies over valuations of the local variables
of thread $k+1$. The notation $\slice{\somestate}{l_{k+1}}$ denotes a
({\format\small$k+1$})-thread state that agrees with $\somestate$ in the
first $k$ local states and whose last local state is $l_{k+1}$; similarly
$\slice{\somestate'}{l_{k+1}'}$. Symbol $\pi$ denotes a permutation on
$\{1,\ldots,k+1\}$ that acts on states by acting on thread indices, which
effectively reorders thread local states.

Asynchronous programs are trivially monotone (and DR):
\equationref{monotone absrel-n} is satisfied by choosing $l_{k+1}' :=
l_{k+1}$ and $\pi$ the identity.
\Tableref{non-monotone} shows instructions found in
\emph{non-}asyn\-chro\-nous programs that destroy monotonicity, and why.
For example, the swap instruction in the first row gives rise to a
\drabb\ program with a 2-thread transition $(0,0,0,0) \in \absrel^2$.
Choosing $l_3 = 1$ in \equationref[]{monotone absrel-n} requires the
existence of a transition in $\absrel^3$ of the form
$(\localvar_1,\localvar_2,\localvar_3,\localvar_1',\localvar_2',\localvar_3')
= (0,0,1,\pi(0,0,\localvar_3'))$, which is impossible: by
equations \equationref[]{n thread DR instantiation} and \equationref[]{DR
  transitions}, there must exist $\idxact \in \{1,2,3\}$ such that for
$\{p,q\} = \{1,2,3\} \setminus \{\idxact\}$, both ``$\idxact$ swaps with
$p$'' and ``$\idxact$ swaps with $q$'' hold, i.e.
\[
  \localvar_p'=\localvar_\idxact \andop \localvar_\idxact'=\localvar_p \wbox{$\andop$} \localvar_q'=\localvar_\idxact \andop \localvar_\idxact'=\localvar_q \mathcomma
\]
which is equivalent to $\localvar_\idxact' = \localvar_p = \localvar_q
\andop \localvar_\idxact = \localvar_p' = \localvar_q'$. It is easy to see
that this formula is inconsistent with the
partial assignment
$(0,0,1,\pi(0,0,\localvar_3'))$, no matter what $\localvar_3'$.
\begin{table}[htbp]
  \begin{centering}
    \tablesizeb{%
      \begin{tabular*}{1\columnwidth}{@{\extracolsep{\fill}}@{\extracolsep{\fill}}@{\extracolsep{\fill}}@{\extracolsep{\fill}}@{\extracolsep{\fill}}@{\extracolsep{\fill}}@{\extracolsep{\fill}}ccccc}
        \toprule
        \multicolumn{2}{c}{\Dr\ program} &  & \multicolumn{2}{c}{Monotonicity}\tabularnewline
        \cmidrule{1-2} \cmidrule{4-5}
        instruction & variables & \hphantom{a} & mon.? & assgn.\ violating \equationref[]{monotone absrel}\tabularnewline
        \midrule
        $\localvar,\localvarpas\assignop\localvarpas,\localvar$ & $\localvar\in\booleans$ &  & \hphantom{a}no\hphantom{a} & $\localvar=0$, $\localvar'=1$ \tabularnewline
        $\localvar,\localvarpas\assignop\localvar+1,\localvarpas-1$ & $\localvar\in\naturals$ &  & yes & \tabularnewline
        $\localvarpas\assignop\localvarpas\plusop\localvar$ & $\localvar\in\naturals$ &  & yes & \tabularnewline
        $\localvar\assignop\localvar\plusop\localvarpas$ & $\localvar\in\naturals$ &  & no & $\localvar=\localvar'=1$ \tabularnewline
        $\localvarpas\assignop c$ & $\localvar,c\in\naturals$ &  & yes & \tabularnewline
        \bottomrule
      \end{tabular*}}
  \end{centering}
  \caption[Examples of monotonicity, and violations of it]{Each row shows a
    single-instruction program, whether the program gives rise to a
    monotone system and, if not,
    an assignment that violates \equationref{monotone absrel}. (Some of
    these programs are not finite-state.)}
  \label{table: non-monotone}
\end{table}

More interesting for us is the fact that asynchronous programs (= our input
language) are monotone, while their parametric predicate abstractions may
not be; this demonstrates that the monotonicity is in fact \emph{lost in
  the abstraction}. Consider again the decrement instruction $\localvar
\assignop \localvar-1$, but this time abstracted against the inter-thread
predicate $\predsym \reldef \localvar = \localvarpas$. Parametric
abstraction results in the two-thread and three-thread template
instantiations
\begin{eqnarray*}
  \parallelinst{\absrel} 2 & = & \left(\negop \localboolvar_1 \orop \negop \localboolvarp_1\right) \ \andop \ \localboolvar_1 = \localboolvar_2                   \ \andop \ \localboolvarp_1 = \localboolvarp_2 \\
  \parallelinst{\absrel} 3 & = & \left(\negop \localboolvar_1 \orop \negop \localboolvarp_1\right) \ \andop \ \localboolvar_1 = \localboolvar_2 = \localboolvar_3 \ \andop \ \localboolvarp_1 = \localboolvarp_2 = \localboolvarp_3 \mathperiod
\end{eqnarray*}
Consider the transition $(0,0) \rightarrow (1,1) \in \parallelinst{\absrel}
2$ and the three-thread state $\otherstate = (0,0,1) \strictlycovers
(0,0)$~: $\otherstate$ clearly has no successor in $\parallelinst{\absrel}
3$ (it is in fact inconsistent),
violating monotonicity. We discuss in \sectionref{Restoring Monotonicity in
  the Abstraction} what happens to the decrement instruction with respect
to predicate $\localvar < \localvarpas$.

\subsection{Restoring Monotonicity in the Abstraction}
\label{section: Restoring Monotonicity in the Abstraction}

Our goal is now to restore the monotonicity that was lost in the parametric
abstraction. The standard covering relation $\coveredby$ defined over local
state counter tuples
turns \emphasize{monotone} and \emphasize{Boolean} \drabb\ programs into
instances of well quasi-ordered transition systems. Program location
reachability is then decidable, even for unbounded threads.

In order to do so, we first derive a sufficient condition for monotonicity
that can be checked \emphasize{locally} over $\absrel$, as follows.
\newcounterset{theoremMonotoneAbsrel}{\theDEF}
\newcommand{\theoremMonotoneAbsrel}[1][]{%
  Let $\absrel$ be the transition relation of a DR program. Then the
  infinite-state transition system $\union_{\threadcount=1}^\infty
  \parallelinst{\absrel}{\threadcount}$ is monotone if the following
  formula over $\abslocalvars \times \abslocalvars'$ is valid:
  \begin{equation}
    \exists \abspvarspas \abspvarspas' \suchthat \absrel \wbox{$\impliesop$} \forall \abspvarspas \exists \abspvarspas' \suchthat \absrel \mathperiod
    \label{equation: monotone absrel#1}
  \end{equation}}
\begin{THE}
  \theoremMonotoneAbsrel
  \label{theorem: monotone absrel}
\end{THE}

Unlike the monotonicity characterization given in \lemmaref{monotone DR
  programs}, \equationref{monotone absrel} is formulated only about the
template program $\absrel$. It suggests that, if $\absrel$ holds for some
valuation of its passive-thread variables, then no matter how we replace
the current-state passive-thread variables $\abspvarspas$, we can find
next-state passive-thread variables $\abspvarspas'$ such that $\absrel$
still holds. This is true for asynchronous programs, since here
$\abslocalvars[\pasindex] = \emptyset$. It fails for the swap instruction
in the first row of \tableref{non-monotone}: the instruction gives rise to
the \drabb\ program $\absrel \reldef \localvar' = \localvarpas \andop
\localvarpas' = \localvar$. The assignment on the right in the table
satisfies $\absrel$, but if $\localvarpas$ is changed to 0, $\absrel$ is
violated no matter what value is assigned to $\localvarpas'$.

\Paragraph

We are now ready to modify the possibly non-monotone abstract
\drabb\ program $\absprogram$ into a new, monotone abstraction
$\absprogrammon$. Our solution is similar in spirit to, but different in
effect from, earlier work on \emph{monotonic abstractions}
\cite{DBLP:journals/entcs/AbdullaDR08}, which proposes to delete processes
that violate universal guards and thus block a transition. This results in
an overappoximation of the original system and thus possibly spuriously
reachable error states. By contrast, exploiting the monotonicity of the
\emph{concrete} program $\program$, we can build a monotone program
$\absprogrammon$ that is safe exactly when $\absprogram$ is,
thus fully preserving soundness and precision of the abstraction
$\absprogram$.
\begin{DEF}
  The \emphdef{non-monotone fragment (NMF)} of a \drabb\ program with
  transition relation $\absrel$ is the formula over $\abslocalvars \times
  \abslocalvars[P] \times \abslocalvars'$:
  \begin{equation}
    \relreset(\absrel) \wbox{\reldef} \neg \exists \abslocalvars[P]' \st \absrel \ \ \andop \ \ \exists \abslocalvars[P] \abslocalvars[P]' \st \absrel \mathperiod
    \label{equation: non-monotone fragment}
  \end{equation}
  \label{definition: non-monotone fragment}
\end{DEF}
The NMF encodes partial assignments $(\localvar,\localvarpas,\localvar')$
that cannot be extended, via any $\localvarpas'$, to a full assignment
satisfying $\absrel$, but can be extended for some valuation of
$\abslocalvars[P]$ other than $\localvarpas$.
We revisit the two non-monotone instructions from \tableref{non-monotone}.
The NMF of $\localvar,\localvarpas \assignop \localvarpas,\localvar$ is
$\localvar' \not= \localvarpas$: this clearly cannot be extended to an
assignment satisfying $\absrel$, but when $\localvarpas$ is changed to
$\localvar'$, we can choose $\localvarpas' = \localvar$ to satisfy
$\absrel$. The non-monotone fragment of $\localvar \assignop \localvar
\plusop \localvarpas$ is $\localvar' \atl \localvar \andop \localvar' \not=
\localvar \plusop \localvarpas$.

\Equationref{non-monotone fragment} is slightly stronger than the negation
of \equationref[]{monotone absrel}: the NMF binds the values of the
$\abslocalvarspas$ variables for which a violation of $\absrel$ is
possible. It can be used to ``repair'' $\absrel$:
\newcounterset{lemmaAddingNonMonotoneFragment}{\theDEF}
\newcommand{\lemmaAddingNonMonotoneFragment}[1][]{%
  For a \drabb\ program with transition relation $\absrel$, the program
  with transition relation $\absrel \orop \relreset(\absrel)$ is monotone.}
\begin{LEM}
  \lemmaAddingNonMonotoneFragment
  \label{lemma: adding non-monotone fragment}
\end{LEM}

\Lemmaref{adding non-monotone fragment} suggests to add artificial
transitions to $\absprogram$ that allow arbitrary passive-thread changes in
states of the non-monotone fragment, thus lifting the blockade previously
caused by some passive threads. While this technique restores monotonicity,
the problem is of course that such arbitrary changes will generally modify
the program behavior; in particular, an added transition may lead a thread
directly into an error state that used to be unreachable.

In order to instead obtain a \emph{safety-equivalent} program,
we prevent passive threads that block a transition in
$\parallelinst{\absprogram}{\threadcount}$ from affecting the future
{\execution}. This can be realized by redirecting them to an auxiliary
sink state.
Let $\labelend$ be a fresh program label.
\begin{DEF}
  The \emphdef{monotone closure} of \drabb\ program $\absprogram =
  (\absrel,\absinitialrel)$ is the \drabb\ program $\absprogrammon =
  (\absrelmon,\absinitialrel)$ with the transition relation $\absrelmon
  \reldef \absrel \orop (\relreset(\absrel) \andop (\pcp' = \labelend))$ .
  \label{definition: monotone closure}
\end{DEF}
\noindent This extension of the transition relation has the following
effects:
\begin{enumerate*}

\item for any program state, if any passive thread can make a move, so can
  all, ensuring monotonicity,

\item the added moves do not affect the safety of the program, and

\item transitions that were previously possible are retained, so no
  behavior is removed.

\end{enumerate*}
The following theorem summarizes these claims:
\newcounterset{theoremSafetyEquivalence}{\theDEF}
\newcommand{\theoremSafetyEquivalence}[1][]{%
  Let $\program$ be an asynchronous program,
  and $\absof\program$ its parametric
  abstraction. The monotone closure $\absprogrammon$ of $\absprogram$ is
  monotone. Further, $\parallelinst{(\absprogrammon)}{\threadcount}$ is
  safe exactly if $\parallelinst{\absprogram}{\threadcount}$ is.}
\begin{THE}
  \theoremSafetyEquivalence
  \label{theorem: safety equivalence}
\end{THE}

\Theoremref{safety equivalence} justifies our strategy for reachability
analysis of an asynchronous program~$\program$: form its parametric
predicate abstraction $\absprogram$ described in
Sections \sectionref[]{Inter-Thread Predicate Abstraction} and
\sectionref[]{From Existential to Parametric Abstraction}, build the
monotone closure $\absprogrammon$, and analyze $(\absprogrammon)^\infty$
using any technique for monotone systems.

\paragraph{Proving the parameterized ticket algorithm}
Applying this strategy to the ticket algorithm yields a well quasi-ordered
transition system for which the backward reachability method described in
\cite{Abdulla10} returns ``uncoverable'', confirming that the ticket
algorithm guarantees mutual exclusion, this time \emph{for arbitrary thread
  counts}.
We remind the reader that the ticket algorithm is challenging for existing
techniques: $\threader$ \cite{GPR11}, $\slab$
\cite{DBLP:conf/tacas/DragerKFW10} and $\symmpa$ \cite{DonaldsonFMSD12}
handle only a fixed number of threads, and the resource requirements of
these algorithms grow rapidly; none of them can handle even a handful of
threads. The recent approach from \cite{FarzanK13} generates
polynomial-size proofs, but again only for fixed thread counts.

\section{Comparison with Related Work}
\label{section: Comparison with Related Work}

Existing approaches for verifying asynchronous shared-memory programs
typically do not exploit the monotone structure that source-level
multi-threaded programs often naturally exhibit~\cite{\concref}.
For example, the constraint-based approach in~\cite{GPR11}, implemented in
$\threader$, generates Owi\-cki-Gries and rely-guarantee type proofs.
It uses predicate abstraction in a CEGAR loop to generate
environment invariants for fixed thread counts, whereas our approach directly
checks the interleaved state space and exploits monotonicity.
Whenever possible, $\threader$
generates thread-modular proofs
by prioritizing
predicates that do not refer to the local variables of other threads. 

A CEGAR approach for fixed-thread symmetric concurrent programs has been
implemented in $\symmpa$~\cite{DonaldsonFMSD12}. It uses predicate
abstraction to generate a Boolean Broadcast program (a special case of
\drabb\ program).
Their approach cannot reason about relationships between local variables
across threads, which is crucial for verifying
algorithms such as the ticket lock. Nevertheless, even the restricted
predicate language of \cite{DonaldsonFMSD12} can give rise to
non-asynchronous programs.
As a result, their technique cannot be extended to unbounded thread counts
with well quasi-ordered systems technology.

Recent work on data flow graph representations of fixed-thread concurrent
programs has been applied to safety property verification~\cite{FarzanK13}.
The inductive data flow graphs can serve as succinct correctness proofs for
safety properties;
for the ticket example they generate correctness proofs of size quadratic
in $\threadcount$.
Similar to \cite{FarzanK13}, the technique in \cite{FarzanK12} uses data flow graphs to
compute invariants of concurrent programs with unbounded threads
(implemented in $\duet$). In contrast to 
our approach,
which uses an
expressive predicate language, $\duet$ constructs proofs from relationships
between either solely shared or solely local variables.
These are insufficient for many benchmarks such as the parameterized ticket algorithm. 

Predicates that, like our inter-thread predicates, reason over all
participating processes/threads have been used extensively in invariant
generation methods \cite{APRXZ01,DBLP:conf/popl/FlanaganQ02,LB04}. As a
recent example, an approach that relies on abstract interpretation instead
of model checking is \cite{SSSC12}. Starting with a set of candidate
invariants (assertions), the approach builds a \emph{reflective
  abstraction},
from which invariants of the concrete system are obtained in a fixed point
process. These approaches and ours share the insight that complex
relationships over all threads may be required to prove easy-to-state
properties such as mutual exclusion. They differ fundamentally in the way
these relationships are used: abstraction with respect to a given set
$\mathcal Q$ of quantified predicates determines the strongest invariant
expressible as a Boolean formula over the set $\mathcal Q$; the result is
unlikely to be expressible in the language that defines $\mathcal Q$.
Future work will investigate how invariant generation procedures can be
used towards \emph{predicate discovery} in our technique.

The idea of ``making'' systems monotone, in order to enable wqo-based
reasoning, was pioneered in earlier
work \cite{BH05,%
DBLP:journals/entcs/AbdullaDR08%
}.
Bingham and Hu deal with guards that require universal quantification over
thread indices, by transforming such systems into Broadcast protocols.
This is achieved by replacing conjunctively guarded
actions
by
transitions that, instead of checking a universal condition,
execute it assuming that any thread not satisfying it ``resigns''. This
happens via a designated local state that isolates such threads from
participation in future the computation.
The same idea was
further developed
by Abdulla et al.\
in the context of
\emph{monotonic abstractions}.
Our solution to the loss of monotonicity was in some way inspired by these
works, but differs in two crucial aspects: first, our concrete input
systems are asynchronous and thus monotone, so our incentive to
\emph{preserve} monotonicity in the abstract is strong. Second, exploiting
the input monotonicity, we can achieve a monotonic abstraction that is
safety-equivalent to the non-monotone abstraction and thus not merely an
error-preserving approximation. This is essential, to avoid spurious
counterexamples in addition to those unavoidably introduced by the
predicate abstraction.

\section{Concluding Remarks}

We have presented in this paper a comprehensive verification method for
arbitrarily-threaded asynchronous shared-variable programs.
Our method is based on predicate abstraction and permits expressive
universally quantified \emph{inter-thread} predicates, which track
relationships such as ``my ticket number is the smallest, among all
threads''. Such predicates are required to verify, via predicate
abstraction, some widely used algorithms like the ticket lock. We found
that the abstractions with respect to these predicates result in
non-monotone finite-data replicated programs, for which reachability is in
fact undecidable. To fix this problem, we strengthened the earlier method
of monotonic abstractions such that it does not introduce spurious errors
into the abstraction.

We view the treatment of monotonicity as the major contribution of this
work. Program design often naturally gives rise to ``monotone
concurrency'', where adding components cannot disable existing actions, up
to component symmetry. Abstractions that interfere with this feature are
limited in usefulness. Our paper shows how the feature can be inexpensively
restored, allowing such abstraction methods and powerful infinite-state
verification methods to coexist peacefully.

\bibliographystyle{abbrv}

\begin{thebibliography}{10}

\bibitem{Abdulla10}
P.~A. Abdulla.
\newblock Well (and better) quasi-ordered transition systems.
\newblock {\em B SYMB LOG}, 2010.

\bibitem{ACJY96}
P.~A. Abdulla, K.~Cerans, B.~Jonsson, and Y.-K. Tsay.
\newblock General decidability theorems of infinite-state systems.
\newblock In {\em LICS}, 1996.

\bibitem{DBLP:journals/entcs/AbdullaDR08}
P.~A. Abdulla, G.~Delzanno, and A.~Rezine.
\newblock Monotonic abstraction in parameterized verification.
\newblock {\em ENTCS}, 2008.

\bibitem{Andrews:1991:CPP:110561}
G.~R. Andrews.
\newblock {\em Concurrent programming: principles and practice}.
\newblock Benjamin-Cummings Publishing Co., Inc., Redwood City, CA, USA, 1991.

\bibitem{APRXZ01}
T.~Arons, A.~Pnueli, S.~Ruah, J.~Xu, and L.~Zuck.
\newblock Parameterized verification with automatically computed inductive
  assertions.
\newblock In {\em CAV}, 2001.

\bibitem{BH05}
J.~D. Bingham and A.~J. Hu.
\newblock Empirically efficient verification for a class of infinite-state
  systems.
\newblock In {\em TACAS}, 2005.

\bibitem{CCKRT06}
S.~Chaki, E.~M. Clarke, N.~Kidd, T.~Reps, and T.~Touili.
\newblock Verifying concurrent message-passing {C} programs with recursive
  calls.
\newblock In {\em TACAS}, 2006.

\bibitem{ClarkeGrumbergLong94}
E.~M. Clarke, O.~Grumberg, and D.~E. Long.
\newblock Model checking and abstraction.
\newblock {\em TOPLAS}, 1994.

\bibitem{CKS07}
B.~Cook, D.~Kroening, and N.~Sharygina.
\newblock Verification of {Boolean} programs with unbounded thread creation.
\newblock {\em Theoretical Comput. Sci.}, 2007.

\bibitem{DonaldsonFMSD12}
A.~F. Donaldson, A.~Kaiser, D.~Kroening, M.~Tautschnig, and T.~Wahl.
\newblock Counterexample-guided abstraction refinement for symmetric concurrent
  programs.
\newblock {\em FMSD}, 2012.

\bibitem{DBLP:conf/tacas/DragerKFW10}
K.~Dr{\"a}ger, A.~Kupriyanov, B.~Finkbeiner, and H.~Wehrheim.
\newblock {SLAB}: A certifying model checker for infinite-state concurrent
  systems.
\newblock In {\em TACAS}, 2010.

\bibitem{FarzanK12}
A.~Farzan and Z.~Kincaid.
\newblock Verification of parameterized concurrent programs by modular
  reasoning about data and control.
\newblock In {\em POPL}, 2012.

\bibitem{FarzanK13CAV}
A.~Farzan and Z.~Kincaid.
\newblock Duet: static analysis for unbounded parallelism.
\newblock In {\em CAV}, 2013.

\bibitem{FarzanK13}
A.~Farzan, Z.~Kincaid, and A.~Podelski.
\newblock Inductive data flow graphs.
\newblock In {\em POPL}, 2013.

\bibitem{FS01}
A.~Finkel and P.~Schnoebelen.
\newblock Well-structured transition systems everywhere!
\newblock {\em Theoretical Comput. Sci.}, 2001.

\bibitem{DBLP:conf/popl/FlanaganQ02}
C.~Flanagan and S.~Qadeer.
\newblock Predicate abstraction for software verification.
\newblock In {\em POPL}, pages 191--202. ACM, 2002.

\bibitem{GS92}
S.~German and P.~Sistla.
\newblock Reasoning about systems with many processes.
\newblock {\em JACM}, 1992.

\bibitem{GrafSaidi97}
S.~Graf and H.~Sa{\"\i}di.
\newblock Construction of abstract state graphs with {PVS}.
\newblock In {\em CAV}. Springer, 1997.

\bibitem{GPR11}
A.~Gupta, C.~Popeea, and A.~Rybalchenko.
\newblock Predicate abstraction and refinement for verifying multi-threaded
  programs.
\newblock In {\em POPL}, 2011.

\bibitem{HJM04}
T.~Henzinger, R.~Jhala, and R.~Majumdar.
\newblock Race checking by context inference.
\newblock In {\em PLDI}, 2004.

\bibitem{KKW12}
A.~Kaiser, D.~Kroening, and T.~Wahl.
\newblock Efficient coverability analysis by proof minimization.
\newblock In {\em CONCUR}, 2012.

\bibitem{LB04}
S.~K. Lahiri and R.~E. Bryant.
\newblock Constructing quantified invariants via predicate abstraction.
\newblock In {\em VMCAI}, volume 2937 of {\em LNCS}, pages 267--281. Springer,
  2004.

\bibitem{malkis10}
A.~Malkis.
\newblock {\em Cartesian Abstraction and Verification of Multithreaded
  Programs}.
\newblock PhD thesis, Albert-Ludwigs-Universit{\"{a}}t Freiburg, 2010.

\bibitem{Minsky:1967:CFI:1095587}
M.~L. Minsky.
\newblock {\em Computation: finite and infinite machines}.
\newblock Prentice-Hall, Inc., Upper Saddle River, NJ, USA, 1967.

\bibitem{SSSC12}
A.~S{\'a}nchez, S.~Sankaranarayanan, C.~S{\'a}nchez, and B.-Y.~E. Chang.
\newblock Invariant generation for parametrized systems using self-reflection.
\newblock In {\em SAS}, pages 146--163, 2012.

\bibitem{DBLP:conf/rp/Schnoebelen10}
P.~Schnoebelen.
\newblock Lossy counter machines decidability cheat sheet.
\newblock In {\em RP}, 2010.

\bibitem{DBLP:conf/cav/TorreMP10}
S.~L. Torre, P.~Madhusudan, and G.~Parlato.
\newblock Model-checking parameterized concurrent programs using linear
  interfaces.
\newblock In {\em CAV}, 2010.

\bibitem{WBKW07}
T.~Witkowski, N.~Blanc, D.~Kroening, and G.~Weissenbacher.
\newblock Model checking concurrent {Linux} device drivers.
\newblock In {\em ASE}, 2007.

\end{thebibliography}

\clearpage

\appendix

\section*{Supplemental Material}

\section{Inter-thread Predicates are Essential for the Ticket Algorithm}
\label{appendix: Inter-thread Predicates are Essential for the Ticket Algorithm}

\begin{LEM}
  Consider the parameterized ticket algorithm where threads call
  \lstinline!spin_lock! and \lstinline!spin_unlock! arbitrarily often. No
  Hoare/Floyd-style correctness proof over \emphdef{single-thread}
  predicates exists.
  \label{lemma: IT needed}
\end{LEM}
\begin{proof}
  We write $\pc_i$ for the pc of thread $i$, $1 \atm i \atm \threadcount$.
  We first state some easy-to-prove invariants of the ticket algorithm:
  \begin{eqnarray}
    & & s \atm t \atm s+\threadcount \\
    & & \pc_i = \ell_1 \impliesop l_i = 0 \\
    & & \pc_i = \ell_2 \impliesop s < l_i < t \\
    & & \pc_i = \ell_3 \impliesop l_i = s \\
    & & \#(\pc = \ell_2) + \#(\pc = \ell_3) = t - s \label{equation: t-s}
  \end{eqnarray}
  We can think of $\ell_1$, $\ell_2$, and $\ell_3$ as the non-critical,
  trying, and locked region of a standard mutex lock. The total number of
  threads in the trying and locked regions is $t-s$ (\equationref{t-s}). If
  all threads are ``non-critical'', we have $s=t$, and the $l_i$ are all
  zero.

  Let now
  \begin{equation}
    \textstyle
    E = \Union_{i=1}^\threadcount E^{s,t,l_i}
  \end{equation}
  be the \emphasize{disjoint} union of sets of predicates formulated over
  the shared variables $s$ and $t$ and any one of the $l_i$ ; in
  particular, no predicate may refer to several of the $l_i$. Suppose $I$
  is an invariant expressible over $E$ that is strong enough to prove
  mutual exclusion. Then
  \begin{equation}
    \forall i,j \st i \not= j \st \ I \andop \pc_i = \pc_j = \ell_2 \wbox{$\impliesop$} l_i \not= l_j \mathcomma
    \label{equation: invariant}
  \end{equation}
  since otherwise threads $i$ and $j$ can, once $s$ reaches the value $l_i$
  ($=l_j$), escape the busy-wait loop and simultaneously proceed to the
  critical section.

  For any $c$, there exists a \emph{reachable} global state satisfying
  $\pc_1=\pc_2=\ell_2$ and $(s,l_1,l_2) = (c,c,c+1)$ (that is, thread 1
  proceeds to the trying region first, then thread 2), and a
  \emph{reachable} global state satisfying $\pc_1=\pc_2=\ell_2$ and
  $(s,l_1,l_2) = (c,c+1,c)$ (vice versa). Since $c$ is unbounded, there
  thus exist infinitely many such assignments that satisfy invariant $I$.

  Let now $\{I_1,\ldots,I_w\}$ be the cubes in the DNF representation of
  $I$. Since this set is finite, there exists a single cube $I_k$ that
  satisfies both $(s,l_1,l_2) = (c,c,c+1)$ and $(s,l_1,l_2) = (c,c+1,c)$,
  for some $c$. We split $I_k$ into the sub-cubes
  that belong to $E^{s,t,l_1}$, and those that belong to $E^{s,t,l_2}$:
  $I_k = I_k^1 \andop I_k^2$; note that these sub-cube sets are disjoint
  (sub-cubes that refer to neither $l_1$ nor $l_2$ are apportioned to
  either side). Then $(s,l_1,l_2) = (c,c,c+1)$ satisfies $I_k^1$, which
  does not contain $l_2$, so in fact $(s,l_1) = (c,c)$ satisfies $I_k^1$.
  Symmetrically, one obtains that $(s,l_2) = (c,c)$ satisfies $I_k^2$.
  Hence $(s,l_1,l_2) = (c,c,c)$ satisfies $I_k^1 \andop I_k^2 = I_k$ and
  hence satisfies $I$, which contradicts \equationref{invariant}.

\end{proof}

\noindent

\section{Simulating Shared Via Local Variables}
\label{appendix: Simulating Shared Via Local Variables}

We have excluded shared variables from the description of \dr\ programs to
simplify the notation. This is not a restriction, as such variables can be
simulated via active- and passive-thread local variables, as follows. To
eliminate shared variable $\sharedvar$, we instead introduce a fresh local
variable $\localvar \in \abspvarsact$, and replace a statement like
$\sharedvar := 5$ by the atomic statement $\localvar := 5, \localvarpas :=
5$. That is, each thread keeps a local copy of what used to be the shared
variable; the semantics of passive-thread variables ensures that the values
are synchronized across all threads.

\section{Proof of \lemmaref{mpa is conservative}}
\label{appendix: Proof of lemma mpa is conservative}

\setcounter{DEF}{\thelemmaMPAisConservative}
\begin{LEM}
  \lemmaMPAisConservative[ repeat]
\end{LEM}
\begin{proof}[partial]
  For the initial states, by equations \equationref[]{existential
    initialrel}, \equationref[]{DR initial} and \equationref[]{ipa} the
  implication amounts to
  \begin{eqnarray*}
    \parallelinst{\exabsinitialrel}{\threadcount} & \impliesop &
    \subsc{\Andop}{\idxpas}[\idxpas\ne\idxact]
    \subst{\subst{\absinitialrel(\threadcount)}{\abslocalvars[\idxact]}{\abspvarsact}}{\abslocalvars[\idxpas]}{\abspvarspas} \\
                                                  & = &
    \subsc{\Andop}{\idxpas}[\idxpas\ne\idxact]
    \exists \exabslocalvars[3],\ldots,\exabslocalvars[\threadcount] \suchthat \
    \parallelinst{\exabsinitialrel}{\threadcount}\{\exabslocalvars[1] \triangleright \abslocalvars\}\{\exabslocalvars[2] \triangleright \abslocalvars[P]\}\{\abslocalvars \triangleright \abslocalvars[\idxact]\}\{\abslocalvars[P] \triangleright \abslocalvars[\idxpas]\} \\
                                                  & = &
    \subsc{\Andop}{\idxpas}[\idxpas\ne\idxact]
    \exists \exabslocalvars[3],\ldots,\exabslocalvars[\threadcount] \suchthat \
    \parallelinst{\exabsinitialrel}{\threadcount}\{\exabslocalvars[1] \triangleright \abslocalvars[\idxact]\}\{\exabslocalvars[2] \triangleright \abslocalvars[\idxpas]\} \mathperiod
  \end{eqnarray*}
  The implication holds since the initial condition $\initialrel$ is
  identical for all threads \equationref[]{n thread instantiation}, so
  replacing thread ids 1 and 2 by thread ids $\idxact$ and $\idxpas$ does
  not falsify the formula. The case of the transition relation is similar
  but more involved.
\end{proof}

\section{Proof of \theoremref{satbound}}
\label{appendix: Proof of theorem: satbound}

\setcounter{DEF}{\thetheoremSatbound}
\begin{THE}
  \theoremSatbound[ repeat]
\end{THE}
\begin{proof}
\global\long\def\sometransition{t}
\global\long\def\somevaluation{v}
\newcommandx\stateperm[1][usedefault, addprefix=\global, 1=]{\pi_{#1}}
\global\long\def\ArelT{{\Arel[\threadcount][\idxact]}}
\global\long\def\ArelTT{{\Arel[\threadcount]}}

Let $\pred 1,\ldots,\pred{\predicatenum}$ be a list of predicates $\numIT$
of which are inter-thread, and let $\absrel_{\infty}$ denote the formula
characterizing $\Orop_{\threadcount=1}^{\infty}\absrel_{\threadcount}$ (the
existence of a finite encoding is guaranteed). We show that stabilization
occurs at $\satbound=2+4\times\numIT$, i.e.,
$\absrel_{\infty}\impliesop\absrel_{\satbound}$. The proof for the
stabilization of $\absinitialrel$ is analogous (factor 4 then reduces to
2). We first show that stabilization occurs for the special case that all
predicates are inter-thread, and then argue that this value is insensitive
to the number of single-thread predicates.

The proof is by induction over $\predicatenum$. Let first
$\predicatenum=\numIT=1$ and
$\sometransition=(\bpredo_{1}\bpredo_{2}\bpredo_{1}'\bpredo_{2}')\in\booleans^{4}$
be a transition in $\absrel_{\infty}$, and $\ArelT\reldef(\bpredact 1
\iffop \predact 1)$. Then by definition of $\absrel_{\infty}$ and
\equationref{rpa} there exists a thread number $\threadcount\ge2$, and a
valuation $\somevaluation$ of variables $\absprogramvars[\threadcount]$,
$\absprogramvars[\threadcount]'$, $\programvars[\threadcount]$ and
$\programvars[\threadcount]'$ such that $\somevaluation$ satisfies
$\parallelinst[1]{\rel}{\threadcount}\andop\Andop_{\idxact=1}^{2}(\bpredo_{\idxact}\iffop\predact
1 \andop\bpredo_{\idxact}'\iffop\predact
1')\andop\Andop_{\idxact=3}^{\threadcount}\ArelT\andop\ArelT'$. Let
\[
  \somevaluation=(\bpredo_{1}\ldots\bpredo_{\threadcount}\bpredo_{1}'\ldots\bpredo_{\threadcount}'\sharedvar\localvar_{1}\ldots\localvar_{\threadcount}\sharedvar'\localvar_{1}'\ldots\localvar_{\threadcount}')
\]
be that valuation. Then there exists a number $q\in[2,6]$, and a map
$\stateperm:\,\{1,\ldots,q\}\to\{1,\ldots,\threadcount\}$ such that
\[
  (\bpredo_{\stateperm[1]}\ldots\bpredo_{\stateperm[\threadcount]}\allowbreak\bpredo_{\stateperm[1]}'\ldots\bpredo_{\stateperm[\threadcount]}'\allowbreak\sharedvar\localvar_{\stateperm[1]}\ldots\localvar_{\stateperm[\threadcount]}\allowbreak\sharedvar'\localvar_{\stateperm[1]}'\ldots\localvar_{\stateperm[\threadcount]}')
\]
satisfies
$\parallelinst[1]{\rel}{\threadcount}\wedge\ArelTT\wedge\ArelTT'$, namely
by defining $\stateperm[1]=1$, $\stateperm[2]=2$, and letting
$\stateperm[3],\ldots,\stateperm[q]$ identify passive threads that falsify
a conjunct in each of the expanded $\pred[1] 1$, $\pred[2] 1$, $\pred[1]
1'$, and $\pred[2] 1'$ (if any)\footnote{Recall that according to
  eq.~\equationref[]{inter-thread semantics}, $\predact 1$ evaluates to
  false whenever the following holds: \\ \format
  $\exists\idxpas\ne\idxact\st\negop{\subst{\subst{\predact
        1}{\localvars[\idxact]}{\localvars}}{\localvars[\idxpas]}{\pvarspas}}$.}.
Then by \equationref{rpa} $\sometransition$ satisfies $\absrel_{q}$ (and
thus $\absrel_{6}$).

For the inductive step from $\predicatenum$ to $\predicatenum+1$ predicates
(all inter-thread), we extend $\stateperm$ by (at most) 4 elements. It
follows that stabilization occurs at $\satbound=2+4\times\numIT$ for any
$\predicatenum=\numIT\ge1$.

It remains to show that stabilization is not thwarted by the existence of
single-thread predicates. By \equationref{inter-thread semantics} it
follows that the truth of such predicates depends only on the variables in
$\programvars[\threadcount]$ that are visible by the thread it is evaluated
over, hence on variables $\programvars[2]$ and $\programvars[2]'$ for any
transition in $\absrel_{\infty}$. Now observe that these values are
maintained in the permutation $(\bpredo_{\stateperm[1]}\ldots)$ defined
above ($\stateperm[j]=j$ for $j \atm 2$), which gives the desired result.
\end{proof}

\section{Proof of \theoremref{undecidability}}
\label{appendix: Proof of theorem: undecidability}

\setcounter{DEF}{\thetheoremUndecidability}
\begin{THE}
  \theoremUndecidability[ repeat]
\end{THE}
\begin{proof}[sketch]
  By reducing the halting problem for the Turing-powerful deterministic
  2-counter Minsky machines \cite{Minsky:1967:CFI:1095587} with $k$ control
  states, to the reachability problem in \drabb\ programs with 3 program
  locations and a
local variable with $k$ values. We demonstrate the reduction using a
deterministic Minsky machine that enumerates pairs in $\naturals^{2}$
(\figureref{Minsky machine and DR program}; the
formalism is from~\cite{DBLP:conf/rp/Schnoebelen10}).

The machine consists of five control states
$\sharedstateo,\ldots,\sharedstateooooo$ ($\sharedstateo$ = initial), two
natural-number counters
$\localstateo$ and $\localstateoo$
(initially $0$), and increment, decrement, and zero-test operations,
denoted by $\MInc{c_i}$, $\MDec{c_i}$ and $\smash{\MTZ{c_i}}$,
\respectivelyend. Each operation changes the control state and counter
value as indicated in the figure (the decrement and zero-test operation
block if $\localstate$ is zero and non-zero, \respectively).

Control states are encoded in local variables of $\absprogram$ ranging over
$\{0,1,2,3,4\}$; as can be seen from the figure, these local variables are
synchronized across the threads, so they simulate a single shared variable
that tracks the control state (see \appendixref{Simulating Shared Via Local
  Variables}). Counters are encoded in program locations $\{d_0,d_1,d_2\}$
of the \drabb\ program $\absprogram$ such that the counter value $c_i$
equals the number of threads in location $d_i$, for $i \in \{1,2\}$.
Location $d_0$ is the single initial program location, thus with an
unbounded number of threads; it merely serves as thread-pool. Control state
changes turn into synchronized local variable updates, together with the
following program counter modifications: for $\MInc{c_i}$ and $\MDec{c_i}$
a thread moves from $d_0$ to $d_i$ and vice versa, respectively, and for
$\smash{\MTZ{c_i}}$ a thread in $d_0$ tests for the absence of a passive
thread in location $d_i$.

Let finally $\ell_e$ be a special program location of $\absprogram$ that is
reached if and only if a local variable has the value that encodes the
Minsky machine's halting state. The machine halts if and only if program
location $\ell_e$ is reached in $\absprogram$.
\global\long\def\somelocalll{*}
\end{proof}
\begin{figure*}
\centering
\begin{adjustbox}{max size={1.0\textwidth}{1.0\textheight}}
\begin{tikzpicture}[minsky]

\node[state] (l2) {$\sharedstateooo$};
\node[accepting,state] (l0) [above left of=l2] {$\sharedstateo$};
\node[state] (l1) [above right of=l2] {$\sharedstateoo$};
\node[state] (l3) [left of=l0] {$\sharedstateoooo$};
\node[state] (l4) [right of=l1] {$\sharedstateooooo$};

\path[->,bend left] (l2) edge node[below left, align=right] {
	$\localvar=\localvar_\pasindex=\sharedstateooo\andop\localvar'=\localvar_\pasindex'=\sharedstateo\andop$\\
	$\pc=d_0\andop\pc'=d_1\andop$\\
	$\pcp'=\pcp$
	} node[above right] {$\MInc{\localstateo}$} (l0);

\path[->,bend left] (l0) edge node[below] {$\ldots$} node[above] {$\MDec{\localstateo}$} (l3);

\path[->,bend left] (l3) edge node[above] {$\ldots$} node[below] {$\MInc{\localstateoo}$} (l0);

\path[->] (l0) edge node[above, align=center] {
	$\localvar=\localvar_\pasindex=\sharedstateo\andop\localvar'=\localvar_\pasindex'=\sharedstateoo\andop$\\
	$\pc=\pc'\andop\pcp=\pcp'\andop$\\
	$\pcp\ne d_1$
} node[below] {$\MTZ{\localstateo}$} (l1);

\path[->,bend right] (l1) edge node[below, align=center] {
	$\localvar=\localvar_\pasindex=\sharedstateoo\andop\localvar'=\localvar_\pasindex'=\sharedstateooooo\andop$\\
	$\pc=d_2\andop\pc'=d_0\andop$\\
	$\pcp'=\pcp$
} node[above] {$\MDec{\localstateoo}$} (l4);

\path[->,bend right] (l4) edge node[above] {$\ldots$} node[below] {$\MInc{\localstateo}$} (l1);

\path[->,bend left] (l1) edge node[below right, align=center] {$\ldots$} node[above left] {$\MTZ{\localstateoo}$}(l2);

\end{tikzpicture}
\end{adjustbox}
\caption{Minsky machine and (part of) its \drabb\ program encoding,
shown as labels of control transitions.
The initial state $\absinitialrel$ of the \drabb\ encoding is
$\localvar=\localvar_\pasindex=\sharedstateo\andop\pc=\pcp=d_0$
}
\label{figure: Minsky machine and DR program}
\end{figure*}

Note how, in the reduction, the zero test affects the passive threads: in
the transition from $\sharedstateo$ to $\sharedstateoo$ in
\figureref{Minsky machine and DR program}, the test on variable
$\localstateo$ is simulated by asserting the absence of a passive thread in
location~$d_1$.

\section{Proof of \lemmaref{monotone DR programs}}
\label{appendix: Proof of lemma: monotone DR programs}

\setcounter{DEF}{\thelemmaMonotoneDRPrograms}
\begin{LEM}
  \lemmaMonotoneDRPrograms[ repeat]
\end{LEM}
\begin{proof}

  \ 

  ``$\Rightarrow$'': suppose $\union_{\threadcount=1}^\infty
  \parallelinst{\absrel}{\threadcount}$ is monotone. Let $\somestate =
  (l_1,\ldots,l_k)$, $\somestate' = (l_1',\ldots,l_k')$ with
  $(\somestate,\somestate') \in \absrel^k$, and $\otherstate =
  \slice{\somestate}{l_{k+1}}$. We have $\somestate \coveredby
  \otherstate$, hence by the monotonicity of
  $\union_{\threadcount=1}^\infty \absrel^\threadcount$ there exists
  $\otherstate'$ such that (a) $(\otherstate,\otherstate') \in
  \union_{\threadcount=1}^\infty \absrel^\threadcount$ and (b) $\somestate'
  \coveredby \otherstate'$. From (a) we conclude that in fact
  $(\otherstate,\otherstate') \in \absrel^{k+1}$. From (b) we conclude that
  $\otherstate'$ contains $k$ threads in local states as in $\somestate'$.
  Let $l_{k+1}'$ be the local state of the additional thread (not
  necessarily the {\format\small $k+1$}st) in $\otherstate'$, and $\sigma$
  be a permutation such that $(\range[]{l_1'}{l_{k+1}'}) =
  \sigma(\otherstate')$. That is, $\sigma$ reorders the local states of
  $\otherstate'$ such that the $k$ local states in $\somestate'$ come
  first, $l_{k+1}'$ comes last. With $\pi := \sigma^{-1}$, we then have
  \begin{eqnarray*}
    \left(\slice{\somestate}{l_{k+1}},\pi(\slice{\somestate'}{l_{k+1}'})\right) & = & \left(\slice{\somestate}{l_{k+1}},\sigma^{-1}(\slice{\somestate'}{l_{k+1}'})\right) \\
                                                                                & = & (\otherstate,\otherstate') \quad \in \quad \parallelinst{\absrel}{k+1} \ .
  \end{eqnarray*}

  ``$\Leftarrow$'': suppose $(\somestate,\somestate') \in
  \union_{\threadcount=1}^\infty \absrel^\threadcount$, say
  $(\somestate,\somestate') \in \absrel^k$, so we write
  $\somestate=(l_1,\ldots,l_k)$ and $\somestate'=(l_1',\ldots,l_k')$. Let
  further $\somestate \coveredby \otherstate$. If $\otherstate$ has $k$
  threads, like $\somestate$, then $\somestate \coveredby \otherstate$
  implies $\somestate \covers \otherstate$: the states are symmetry
  equivalent, say $\otherstate = \pi(\somestate)$, for a permutation $\pi$
  on $\range 1 k$. In this case $\otherstate' := \pi(\somestate')$
  satisfies the monotonicity conditions.

  If $\otherstate$ has $k+1$ threads, then observe that $\otherstate$
  contains $k$ threads in local states as in $\somestate$; let $l_{k+1}$ be
  the local state of the additional thread (not necessarily the
  {\format\small $k+1$}st). Let further $l_{k+1}'$ and $\pi$ be as provided
  in \equationref[]{monotone absrel-n repeat}. With $\yetotherstate =
  \slice{\somestate}{l_{k+1}}$ and $\yetotherstate' =
  \pi(\slice{\somestate'}{l_{k+1}'})$, we get
  $(\yetotherstate,\yetotherstate') \in \absrel^{k+1}$ by
  \equationref[]{monotone absrel-n repeat}. Since $\yetotherstate$ and
  $\otherstate$ contain the same local states, let $\sigma$ be a
  permutation such that $\sigma(\yetotherstate) = \otherstate$. Define
  $\otherstate' = \sigma(\yetotherstate')$. Then $\otherstate' \sim
  \yetotherstate' = \pi(\slice{\somestate'}{l_{k+1}'}) \covers
  \somestate'$, where $\sim$ is symmetry equivalence. Further,
  $(\yetotherstate,\yetotherstate') \in \absrel^{k+1}$ implies
  $(\sigma(\yetotherstate),\sigma(\yetotherstate')) \in \absrel^{k+1}$ by
  symmetry, so $(\otherstate,\otherstate') \in \absrel^{k+1} \subseteq
  \union_{\threadcount=1}^\infty \absrel^\threadcount$, demonstrating that
  the monotonicity conditions are satisfied.

  The case that $\otherstate$ has more than $k+1$ threads follows by
  induction.
\end{proof}

\section{Proof of \theoremref{monotone absrel}}
\label{appendix: Proof of theorem: monotone absrel}

\setcounter{DEF}{\thetheoremMonotoneAbsrel}
\begin{THE}
  \theoremMonotoneAbsrel[ repeat]
\end{THE}
\begin{proof}
  We show monotonicity using \lemmaref{monotone DR programs}. Suppose
  $(\somestate,\somestate') \in \absrel^k$, and let $l_{k+1}$ be given. By
  \equationref[]{n thread DR instantiation}, there exists $\idxact \in
  \{1,\ldots,k\}$ such that $(\somestate,\somestate') \in
  \parallelinst[\idxact]{\absrel}{k}$. By \equationref[]{DR transitions},
  we have
  \begin{equation}
    \forall \idxpas \in \{1,\ldots,k\} \setminus \{\idxact\} \ \substp{\substp{\absrel}{\abslocalvars[\idxact]}{\abspvarsact}}{\abslocalvars[\idxpas]}{\abspvarspas} \mathperiod
    \label{equation: instantiated lemma}
  \end{equation}
  Since $k \atl 2$, the quantification in \equationref[]{instantiated
    lemma} is not empty and hence satisfies the left-hand side of
  \equationref[]{monotone absrel repeat}. By the right-hand side, there
  exists a valuation $l_{k+1}'$ of all $\abspvarspas'$ variables such that,
  replacing the $\abspvarspas$ variables by the valuation $l_{k+1}$,
  $\absrel$ still holds,
  i.e.\ $\substp{\substp{\absrel}{\abslocalvars[\idxact]}{\abspvarsact}}{\abslocalvars[k+1]}{\abspvarspas}$.
  Merging this with \equationref[]{instantiated lemma}, we obtain
  \[
    \forall \idxpas \in \{1,\ldots,k+1\} \setminus \{\idxact\} \ \substp{\substp{\absrel}{\abslocalvars[\idxact]}{\abspvarsact}}{\abslocalvars[\idxpas]}{\abspvarspas} \mathcomma
  \]
  and thus $(\slice{\somestate}{l_{k+1}},\slice{\somestate',l_{k+1}'}) \in
  \parallelinst[\idxact]{\absrel}{k+1} \subset \absrel^{k+1}$, establishing
  the right-hand side of \equationref[]{monotone absrel-n repeat} with the
  identity permutation $\pi$.
\end{proof}

\section{Proof of \lemmaref{adding non-monotone fragment}}
\label{appendix: Proof of lemma: adding non-monotone fragment}

\setcounter{DEF}{\thelemmaAddingNonMonotoneFragment}
\begin{LEM}
  \lemmaAddingNonMonotoneFragment[ repeat]
\end{LEM}
\begin{proof}
  We show that $\absrel \orop \relreset(\absrel)$ satisfies
  \equationref[]{monotone absrel repeat}, i.e.
  \begin{equation}
    \exists \abspvarspas \abspvarspas' \suchthat (\absrel \orop \relreset(\absrel)) \wbox{$\impliesop$} \forall \abslocalvars[P] \exists \abslocalvars[P]' \st (\absrel \orop \relreset(\absrel))
    \label{equation: proof obligation}
  \end{equation}
  Monotonicity then follows using \theoremref{monotone absrel}.

  We first simplify the right-hand side of \equationref[]{proof
    obligation}:
  \[
  \begin{array}{rl}
      & \forall \abslocalvars[P] \exists \abslocalvars[P]' \st (\absrel \orop (\negop \exists \abslocalvars[P]' \st \absrel \; \andop \; \exists \abslocalvars[P] \abslocalvars[P]' \st \absrel)) \\
    = & \forall \abslocalvars[P] \st (\exists \abslocalvars[P]' \st \absrel \orop (\negop \exists \abslocalvars[P]' \st \absrel \; \andop \; \exists \abslocalvars[P] \abslocalvars[P]' \st \absrel)) \\
    = & \forall \abslocalvars[P] \st (\exists \abslocalvars[P]' \st \absrel \orop \exists \abslocalvars[P] \abslocalvars[P]' \st \absrel) \\
    = & \forall \abslocalvars[P] \st (\exists \abslocalvars[P] \abslocalvars[P]' \st \absrel) \\
    = & \exists \abslocalvars[P] \abslocalvars[P]' \st \absrel \mathperiod
  \end{array}
  \]
  \Equationref{proof obligation} now becomes
  \[
    \exists \abspvarspas \abspvarspas' \suchthat (\absrel \orop (\negop \exists \abslocalvars[P]' \st \absrel \; \andop \; \exists \abslocalvars[P] \abslocalvars[P]' \st \absrel)) \wbox{$\impliesop$} \exists \abslocalvars[P] \abslocalvars[P]' \st \absrel
  \]
  which trivially reduces to $\true$, in both cases of the disjunction.
\end{proof}

\section{Proof of \theoremref{safety equivalence}}
\label{appendix: Proof of theorem: safety equivalence}

\setcounter{DEF}{\thetheoremSafetyEquivalence}
\begin{THE}
  \theoremSafetyEquivalence[ repeat]
\end{THE}
\begin{proof}
  \ 
  \begin{enumerate}

  \item Monotonicity of $\absprogrammon$: employing \theoremref{monotone
    absrel}, we prove that the following formula is valid:
    \[
      \exists \abspvarspas \abspvarspas' \suchthat \absrelmon \wbox{$\impliesop$} \forall \abspvarspas \exists \abspvarspas' \suchthat \absrelmon \mathperiod
    \]
    Let $(\localvar,\localvar') \in \abslocalvars \times \abslocalvars'$ be
    arbitrary, and suppose there exist $(\localvarpas,\localvarpas') \in
    \abslocalvarspas \times \abslocalvarspas'$ such that
    $(\localvar,\localvarpas,\localvar',\localvarpas') \in \absrelmon$. Let
    further $\localvarotherpas \in \abslocalvarspas$. We construct
    $\localvarotherpas' \in \abslocalvarspas'$ such that
    $(\localvar,\localvarotherpas,\localvar',\localvarotherpas') \in
    \absrelmon$.

    Since $\absrelmon = \absrel \orop (\relreset(\absrel) \andop (\pcp' =
    \labelend))$, we have either
    $(\localvar,\localvarpas,\localvar',\localvarpas') \in \absrel$ or
    $(\localvar,\localvarpas,\localvar',\localvarpas') \in
    \relreset(\absrel) \andop (\pcp' = \labelend)$. In both cases,
    $(\localvar,\localvarpas,\localvar',\localvarpas') \in \absrel \orop
    \relreset(\absrel)$. The latter relation is monotone by
    \lemmaref{adding non-monotone fragment}. Hence there exists some
    $\localvaryetotherpas'$ such that
    $(\localvar,\localvarotherpas,\localvar',\localvaryetotherpas') \in
    \absrel \orop \relreset(\absrel)$.

    This element $\localvaryetotherpas'$ is almost the element
    $\localvarotherpas' \in \abslocalvarspas'$ we are looking for: if we
    have $(\localvar,\localvarotherpas,\localvar',\localvaryetotherpas')
    \in \absrel \subset \absrelmon$, then the choice $\localvarotherpas' =
    \localvaryetotherpas'$ ensures
    $(\localvar,\localvarotherpas,\localvar',\localvarotherpas') \in
    \absrelmon$. Otherwise
    $(\localvar,\localvarotherpas,\localvar',\localvaryetotherpas') \in
    \relreset(\absrel)$. Formula $\relreset(\absrel)$ does not contain
    $\abslocalvarspas'$ variables, however; the latter can thus be replaced
    freely without affecting membership in $\relreset(\absrel)$. Let
    therefore $\localvarotherpas' \reldef (\exists \pcp'
    \localvaryetotherpas') \andop (\pcp' = \labelend)$. The latter
    expression denotes the replacement of the value of $\pcp'$ in
    $\localvaryetotherpas'$ by $\labelend$. Now we have
    $(\localvar,\localvarotherpas,\localvar',\localvarotherpas') \in
    \relreset(\absrel)$ and in fact
    $(\localvar,\localvarotherpas,\localvar',\localvarotherpas') \in
    \relreset(\absrel) \andop (\pcp' = \labelend) \subset \absrelmon$.

  \item Safety equivalence: from \definitionref{monotone closure} (applied
    to $\absprogram$) we
conclude $\absrel \impliesop \absrelmon$ is valid, and thus
every execution of $\absprogram$ is an execution of $\absprogrammon$.
Thus if $\absprogrammon$ is safe, so is~$\absprogram$.

For the converse argument observe that every infinite trace $\pi$ of
$\absprogrammon$ gives rise to a sequence of $j$ traces of $\absprogram$ as
follows:
\[
\pi = t_1,\ldots,r_1,t_2,\ldots,r_2,\ldots,t_j,\ldots
\]
such that
for all~$i$, subtrace $t_i,\ldots,r_i$ is pairwise related by~$\absrel$,
$(r_i,t_{i+1}) \notin \absrel$, yet $(r_i,t_{i+1}) \in \absrelmon$.
(If~$\pi$ is finite it is of the form~$t_1,\ldots,r_1,\ldots,t_j,\ldots,r_j$; the following remains valid.)

Call a state \inlinedef{safe} if it has no emanating execution ending in an
error state. Since the asynchronous input program~$\program$ is monotone
(``fewer threads can do less''), state safety
is~$\strictlycoveredby$-closed for~$\program$: if a state~$r$ is safe
in~$\program$ and~$s \strictlycoveredby r$ then~$s$ is also safe. In order
to see that the same is true for the (possibly non-monotone) abstract
\drabb\ program~$\absprogram$, let~$R$ be the concretization of a state~$r$
of~$\absprogram$, \ie\ a set of programs states of input
program~$\program$. Then $\absprogram$'s conservativeness
(\sectionref{Existential Inter-Thread Predicate Abstraction} and
\corollaryref{saturation}) guarantee the safety of states in~$R$, and
$\strictlycoveredby$-closedness of state safety in~$\program$ implies the
safety of states in the~$\strictlycoveredby$-downward closure of~$R$. From
the fact that~$s$'s concretization is in that closure we can conclude that
state safety is also~$\strictlycoveredby$-downward closed
for~$\absprogram$.

Using the previous result we next show that if a subtrace $t_i,\ldots,r_i$
of~$\pi$ contains no error state, then neither does
$t_{i+1},\ldots,r_{i+1}$; induction then gives us the desired result.
$t_1,\ldots,r_1$ contains no error state (otherwise~$\absprogram$ cannot be
safe). The proof of the induction step is by contradiction. Assume
$t_i,\ldots,r_i$ contains no error state, yet $t_{i+1},\ldots,r_{i+1}$ does
so. Let~$r_i'$ be a state such that $r_i' \strictlycoveredby r_i$ and
$(r_i',t_{i+1})\in\absrel$. Such a state is always guaranteed to exist.%
\footnote{Such a state can always be obtained from~$r_i$ by removing the threads that were redirected to an auxiliary
state in transition $(r_i,t_{i+1})\in\absrelmon$.
}
Hence
$r_i$ is safe,
$r_i' \strictlycoveredby r_i$, yet
$r_i'$ not safe,
which contradicts the property that state safety is~$\strictlycoveredby$-closed for~$\absprogram$ and gives the desired result.

\end{enumerate}

\end{proof}

\section{Backward Reachability Tree for the Ticket Algorithm}
\label{appendix: Backward Reachability Tree for the Ticket Algorithm}

\Figureref{ticketproof} shows the backward reachability tree for the Ticket
algorithm obtained using the $\breach$ infinite-state model checker
\cite{KKW12}.
\begin{figure}[ht]
  \centering
  \begin{scaletikzpicturetowidth}{0.96\columnwidth}
    \begin{tikzpicture}[reachtrees]

      \node (s1l25l25) at (341bp,241bp+74bp) [bwnode,minstate,align=center] {$\boldsymbol{\LOCALSTATE{25}}$\\$\boldsymbol{\LOCALSTATE{25}}$};

      \node (s0l18) at (247bp,241bp) [bwnode,minstate,align=center] {$\LOCALSTATE{18}$};
      \node (s0l20) at (247bp-94bp,241bp) [bwnode,minstate,align=center] {$\LOCALSTATE{20}$};
      \node (s0l16) at (435bp,241bp) [bwnode,minstate,align=center] {$\LOCALSTATE{16}$};
      \node (s0l22) at (435bp+94bp,241bp) [bwnode,minstate,align=center] {$\LOCALSTATE{22}$};

      \node (s0l19) at (341bp,241bp) [bwnode,minstate,align=center] {$\LOCALSTATE{19}$};
      \node (s0l1) at (220bp,19bp) [bwnode,minstate,align=center] {$\LOCALSTATE{1}$};
      \node (s0l21) at (435bp,167bp) [bwnode,minstate,align=center] {$\LOCALSTATE{21}$};
      \node (s0l11) at (341bp,167bp) [bwnode,minstate,align=center] {$\LOCALSTATE{11}$};
      \node (s0l10) at (341bp,93bp) [bwnode,minstate,align=center] {$\LOCALSTATE{10}$};
      \node (s0l13) at (435bp,93bp) [bwnode,minstate,align=center] {$\LOCALSTATE{13}$};
      \node (s0l0) at (302bp,19bp) [bwnode,minstate,align=center] {$\LOCALSTATE{0}$};
      \node (s0l5) at (466bp,19bp) [bwnode,minstate,align=center] {$\LOCALSTATE{5}$};
      \node (s0l4) at (384bp,19bp) [bwnode,minstate,align=center] {$\LOCALSTATE{4}$};
      \node (s0l9) at (250bp,93bp) [bwnode,minstate,align=center] {$\LOCALSTATE{9}$};
      \node (s0l17) at (250bp,167bp) [bwnode,minstate,align=center] {$\LOCALSTATE{17}$};
      \draw [covprereledge]       (s0l13) -- (s0l5)  ;
      \draw [covprereledge]       (s0l10) -- (s0l0)  ;
      \draw [covprereledge]       (s0l21) -- (s0l13) ;
      \draw [covprereledge]       (s0l10) -- (s0l4)  ;
      \draw [covprereledge]       (s0l13) -- (s0l0)  ;
      \draw [covprereledge]       (s0l19) -- (s0l11) ;
      \draw [covprereledge]       (s0l10) -- (s0l1)  ;
      \draw [covprereledge]       (s0l9)  -- (s0l0)  ;
      \draw [covprereledge]       (s0l13) -- (s0l4)  ;
      \draw [covprereledge]       (s0l17) -- (s0l9)  ;
      \draw [covprereledge]       (s0l10) -- (s0l5)  ;
      \draw [covprereledge]       (s0l9)  -- (s0l1)  ;
      \draw [covprereledge]       (s0l9)  -- (s0l4)  ;
    \end{tikzpicture}
  \end{scaletikzpicturetowidth}
  \caption[Minimal uncoverability proof for the ticket algorithm]{The algorithm
    used \cite{KKW12} attempts to prove uncoverability of \emph{smaller}
    ($\strictlycoveredby$) undecided elements first, which is why some
    (larger) elements are not expanded}
  \label{figure: ticketproof}
\end{figure}
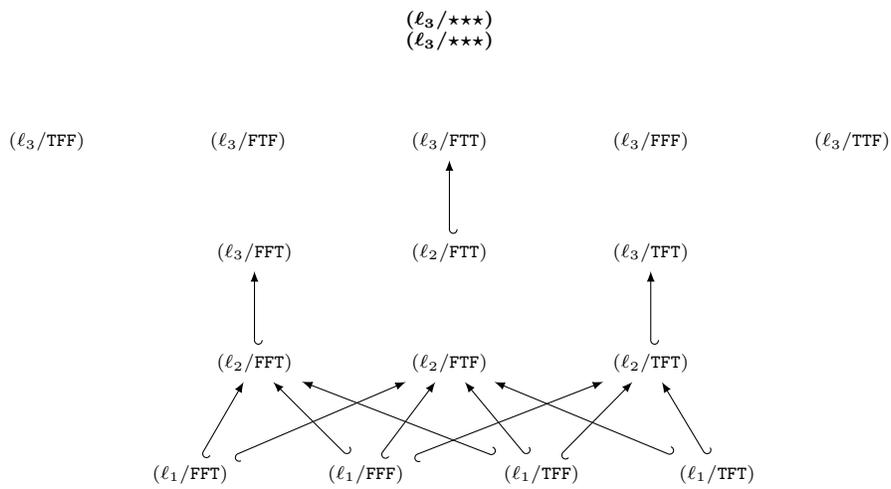

\end{document}